\documentclass[sigconf, nonacm]{acmart}
\usepackage{mathtools}
\usepackage{amssymb}
\usepackage{listings}
\usepackage{array}
\usepackage{xspace} 
\usepackage[pdf]{graphviz}
\usepackage[normalem]{ulem} 
\usepackage{comment}
\usepackage{graphicx}
\graphicspath{ {./figures/} }
\usepackage{varwidth}
\usepackage[ruled, linesnumbered]{algorithm2e}
\usepackage{enumitem}
\usepackage{footnote}
\usepackage{multirow}
\usepackage{tablefootnote}
\usepackage{xpatch}
\usepackage{algorithm2e}
\newcommand{\forfullversion}[1]{}
\makeatletter
\newcommand*{\addFileDependency}[1]{
  \typeout{(#1)}
  \@addtofilelist{#1}
  \IfFileExists{#1}{}{\typeout{No file #1.}}
}
\makeatother
\xpretocmd{\digraph}{\addFileDependency{#2.dot}}{}{}

\setlist{  
  listparindent=\parindent,
  parsep=0pt,
}

\newcommand{\sys}{\textsf{\textsc{Janus}AQP}\xspace}
\newcommand{\syn}{\textsf{DPT}\xspace}
\newcommand{\prior}{\textsf{SPT}\xspace}

\newcommand{\code}[1]{\texttt{#1}}
\long\def\comment#1{}

\newcommand{\polylog}{\textrm{polylog}}
\newcommand{\remove}[1]{}
\def\mparagraph#1{\par\medskip\noindent{\textbf{#1.}}\quad\parindent 1.5em}

\newcommand{\MAXV}{\mathcal{M}}
\newcommand{\fullData}{\mathcal{D}}
\newcommand{\samples}{S}
\newcommand{\PASStree}{\mathcal{T}}
\newcommand{\maxVar}{\mathcal{M}}
\newcommand{\MAXVar}{\mathcal{M}}
\newcommand{\optVar}{\mathcal{V}}

\newcommand\vldbdoi{XX.XX/XXX.XX}

\newcommand\vldbavailabilityurl{URL_TO_YOUR_ARTIFACTS}
\newcommand\vldbpagestyle{plain}

\newcommand{\newchanges}[1]{#1}
\newcommand{\reviewone}[1]{#1}
\newcommand{\reviewtwo}[1]{#1}

\begin{document}

\title{\sys: Efficient Partition Tree Maintenance for Dynamic Approximate Query Processing}

\lstset{%
    basicstyle=\ttfamily,
    backgroundcolor=\color{white},
    breakatwhitespace=false,
    breaklines=true,
    captionpos=b,
    commentstyle=\color{dkgreen},
    escapeinside={\%*}{*)},
    extendedchars=true,
    frame=none,
    keepspaces=true,
    keywordstyle=\color{blue},
    language=SQL,
    numbers=none,
    showspaces=false,
    showstringspaces=false,
    showtabs=false,
    aboveskip=5pt,
    belowskip=5pt,
    sensitive=false
}

\author{Xi Liang}
\affiliation{%
  \institution{University of Chicago}
}
\email{xiliang@uchicago.edu}

\author{Stavros Sintos}
\affiliation{%
  \institution{University of Chicago}
}
\email{sintos@uchicago.edu}

\author{Sanjay Krishnan}
\affiliation{%
  \institution{University of Chicago}
}
\email{skr@cs.uchicago.edu}


\begin{abstract}
Approximate query processing over dynamic databases, i.e., under insertions/deletions, has applications ranging from high-frequency trading to internet-of-things analytics.
We present \sys, a new dynamic AQP system, which supports SUM, COUNT, AVG, MIN, and MAX queries under insertions and deletions to the dataset. 
\sys extends static partition tree synopses, which are hierarchical aggregations of datasets, into the dynamic setting.
This paper contributes new methods for: (1) efficient initialization of the data synopsis in the presence of incoming data, (2) maintenance of the data synopsis under insertions/deletions, and (3) re-optimization of the partitioning to reduce the approximation error. 
\sys reduces the error of a state-of-the-art baseline by more than 60\% using only 10\% storage cost. \sys can process more than 100K updates per second in a single node setting and keep the query latency at a millisecond level.


\end{abstract}
\maketitle

\pagestyle{\vldbpagestyle}


\setcounter{section}{0}
\section{Introduction}
Approximate query processing (AQP) studies principled ways to sacrifice query result accuracy for faster or more resource-efficient execution~\cite{chaudhuri2017approximate, garofalakis2001approximate}.
AQP systems generally employ reduced-size summaries, or ``synopses'', of large datasets that are faster to process.
The simplest of such synopsis structures are histograms and samples~\cite{cormode2011synopses, liang2021combining, agarwal2013blinkdb, lazaridis2001progressive}, but many others have been proposed in the literature.
More complex synopses are more accurate for specific types of queries~\cite{walenz2019learning}, specific data settings~\cite{poepselland}, or even are learned with machine learning models~\cite{yang2019deep, hilprecht2019deepdb, ma2021learned}. 
AQP is particularly interesting and challenging in a  \textbf{dynamic data setting}, where a dataset is continuously modified with insertions and deletions~\cite{garofalakis2001approximate, olma2019taster, acharya1999aqua}.
In this setting, hereafter denoted as DAQP, any synopsis data structures have to be continuously maintained online.

As an example use-case, consider a database aggregating per-stock order data for the NASDAQ exchange~\cite{nasdaqbv}.
Suppose, that we would like to build a low-latency SQL interface for approximate aggregate queries over the past seven days of order data. 
On a typical day, there are 25M new orders that correspond to trades that are placed by brokers (up to 70,000 orders in any given second).
A decent fraction of these orders are eventually canceled or prematurely terminated, for a variety of financial reasons.
Thus, this database is highly dynamic with a large volume of new insertions (new orders) and a small but significant number of deletions (canceled orders).
This paper explores such scenarios with similar motivating applications in internet-of-things monitoring and enterprise stream processing.

Simple synopses like 1D histograms and uniform samples are easy to maintain dynamically. However, such structures are often inaccurate in high-dimensional data and selective query workloads. More complex synopses structures, e.g, ~\cite{yang2019deep,liang2021combining} can be optimized for a particular instance (dataset and query workload),  but are generally harder to maintain online.
For example, recently proposed learned synopses require expensive retraining procedures which limit insertion/deletion throughput~\cite{yang2019deep, hilprecht2019deepdb, ma2021learned}.
Even classical stratified samples may have to be periodically re-optimized and re-balanced based on query and workload shifts~\cite{agarwal2013blinkdb}.
These, expensive (re-)initialization procedures can significantly hurt insertion throughput, and accordingly, almost all existing AQP systems focus on the static data warehousing setting\footnote{A notable exception being the AQUA project~\cite{acharya1999aqua} from 20 years ago.}.
Unfortunately, the existing techniques that \emph{are} designed for dynamic data, such as sketches and mergeable summaries~\cite{gan2020coopstore, agarwal2012mergeable,poepselland, cormode2011synopses, gan2020coopstore, agarwal2012mergeable}, often cannot handle arbitrary deletions or aggregation queries with arbitrary predicates easily.

\reviewone{
In particular, traditional sketch techniques~\cite{cormode2011sketch, cormode2011synopses} focus on frequency, distinct values, and quantile problems. They do not handle arbitrary range queries with low error guarantees and they only consider updates in the streaming setting, i.e. they do not handle arbitrary insertions and deletions.
Our system has a few similarities with high dimensional histograms which is another type of sketching that has been used for aggregation queries. We construct a space partitioning and we store some statistics to answer queries, which is also the case for high dimensional histograms. However, we propose a dynamic method that combines partitioning (histograms) with stratified sampling and extend it to handle dynamic updates efficiently. Multi-dimensional histograms usually cannot be updated easily under arbitrary insertions and deletions, for example \cite{thaper2002dynamic, gilbert2002fast} do not handle general updates, and \cite{donjerkovicdynamic} does not guarantee optimality. Furthermore, most histograms focus on equi-depth, equi-width or V-optimal partitioning.
We maintain a partitioning that guarantees that the confidence intervals we get are close to the optimum (shortest) confidence intervals we could have from the best partitioning. 
}

\reviewone{
In addition, mergeable summaries~\cite{agarwal2012mergeable, gan2020coopstore} focus mostly on frequency and quantiles problems.
Traditional mergeable summaries for aggregation queries mainly include (uniform, stratified, reservoir) sampling. While sampling techniques support efficient updates in the database, the error is far from optimum.
Our method is essentially a technique that uses mergeable summaries for aggregation queries. However, in contrast to traditional mergeable summaries, our new system uses a more involved technique: Maintain a partitioning over the space and a stratified sampling under arbitrary insertions and deletions. When we get a query we combine low-error (updated) summaries with only a few (updated) higher-error summaries to return a better estimation than traditional mergeable summaries techniques.
}

Thus, it is understood that most synopsis data structures have at least one of the following pitfalls in our desired dynamic setting: throughput, drift, or generality~\cite{chaudhuri2017approximate}.

\begin{figure}[t]
    \centering
    \includegraphics[scale=0.25]{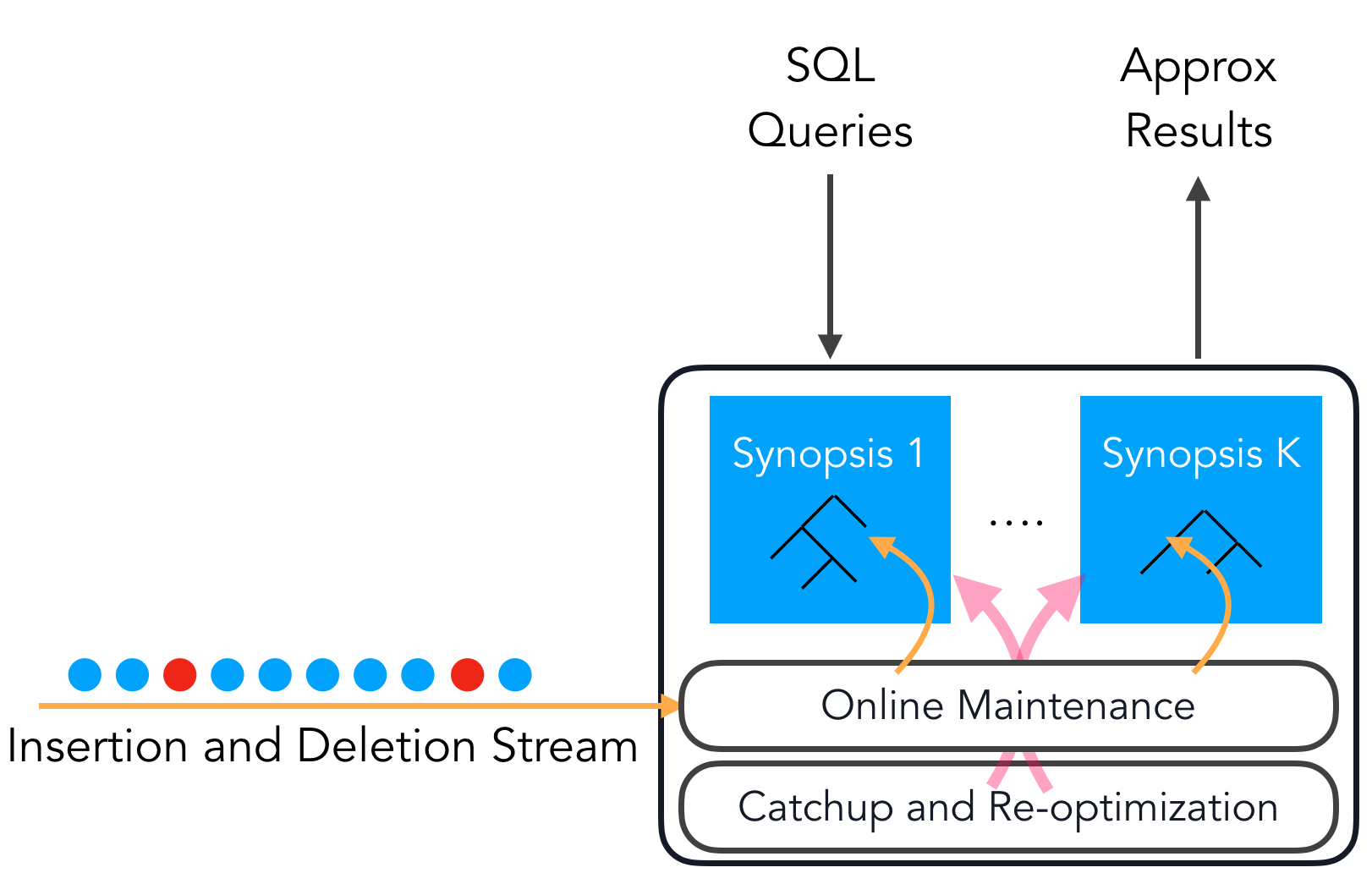}
    \caption{\sys manages a collection of \syn synopses by maintaining them online while periodically re-optimizing partitioning and sample allocation. \label{fig:overview}}\vspace{-0.5em}
\end{figure}

This paper explores the DAQP problem and studies ways that we can mitigate the pitfalls of prior approaches with a flexible synopsis data structure that can continuously re-optimize itself.
We present \sys, a new DAQP system, which supports SUM, COUNT, AVG, MIN, and MAX queries with predicates under arbitrary insertions and deletions to the dataset. 
The main data structure in \sys is a dynamic extension of our recently published work~\cite{liang2020fast,liang2021combining}, which we call a Dynamic Partition Tree (\syn).
\syn is a two-layer synopsis structure that consists of a: (1) hierarchical partitioning of a dataset into a tree, and (2) a uniform sample of data for each of the leaf partitions (effectively a stratified sample over the leaves).
An optimizer determines the best partitioning conditions and sample allocations to meet a user's performance goals.
For each partition (nodes in the tree), we calculate the SUM, COUNT, MIN, and MAX values of the partition.
Any desired SUM, COUNT, AVG, MIN, and MAX query can be efficiently decomposed into two parts with the structure: a combination of the partial aggregates where the predicate fully covers a partition in the tree, and an approximate part where the predicate partially covers a leaf node (and can be estimated with a sample).
More importantly, this structure is essentially a collection of materialized views and samples, which can be maintained incrementally.

A core contribution of \sys is online synopsis optimization.
\sys continuously monitors the accuracy of all of its \syn synopses to account for data and workload drift. 
When a synopsis is no longer accurate, it triggers a re-optimization procedure that resamples and repartitions the data. 
This re-optimization problem is both a significant algorithmic and systems challenge.
From an algorithmic perspective, \sys needs an efficient way to determine the optimal partitioning conditions in dynamic data.
We propose an efficient algorithm based on a dynamic range tree index that finds a partitioning that controls the minimax query error (up to an approximation factor).
From a systems perspective,  re-optimization poses a bit of a logistical challenge.
New data will arrive as the new synopsis data structure is being constructed. We design an efficient multi-threaded catch-up processing algorithm that synchronizes new data and historical data without sacrificing the statistical rigor of the estimates. 

\reviewone{
Overall, the new system \sys we proposed has the following benefits over previously known indexes for approximate query processing. It handles arbitrary dynamic updates efficiently (comparing to the static indexes PASS~\cite{liang2021combining}, VecrdictDB~\cite{park2018verdictdb}, or some dynamic histograms~\cite{gilbert2002fast} that do not handle arbitrary updates), it provides theoretical guarantees on the confidence intervals (comparing to the machine learning based indexes such that DeepDB~\cite{hilprecht2019deepdb}), the query procedure accesses only a small synopsis of data and does not touch the original data set so the communication throughput is low (comparing to other tree-based indexes such as~\cite{joshi2008materialized, jurgens1998r}), and the estimation error is always low (comparing to traditional sketches~\cite{cormode2011synopses}, or histograms for aggregation queries)  without making any assumption about the spatial/value-domain distribution of the data (comparing to~\cite{lazaridis2001progressive}). Furthermore, our new system maintains the strongest theoretical guarantees on the error comparing to all the other known methods we discussed. In particular, the confidence intervals we get are close to the optimum confidence intervals we could have from the best partitioning.}
\section{Background}
We first introduce the core concepts behind the synopses used in this work. 

\subsection{Dynamic Approximate Query Processing}
We assume an initial database table $\fullData^{(0)}$.
This table $\fullData^{(0)}$ is continuously modified through a stream of insertions and deletions of tuples.
As a design principle, we assume that insertions are common but deletions are rare.
With each insertion or deletion operation, the table evolves over time with a new \textbf{state} at each time step $i$:
$\fullData^{(0)},\fullData^{(1)},\ldots,\fullData^{(i)}, \fullData^{(i+1)},\ldots$
A \textbf{synopsis} is a data structure that summarizes the evolving table. 
For each $\fullData^{(i)}$, there is a corresponding synopsis $\Sigma^{(i)}$:
$\Sigma^{(0)},\Sigma^{(1)},\ldots,\Sigma^{(i)}, \Sigma^{(i+1)},\ldots$

In DAQP, the problem is to answer queries as best as possible from only the $\Sigma^{(i)}$.
For a query $q$, the estimation error is defined as the difference between the estimated result (using the synopsis) and the true result (using the current database state): 
\[ \textsf{Error}(q, \Sigma^{(i)}) = |q(\fullData^{(i)}) ~~ -  ~~ q(\Sigma^{(i)})|\]
We further assume that there is sufficient cold/archival storage to store the current state of the table $\fullData^{(i)}$.
This data can be accessed in an offline way for initialization, re-optimization, and logging purposes but not for query processing.

\begin{figure}[t]
    \centering
    \includegraphics[width=1\columnwidth]{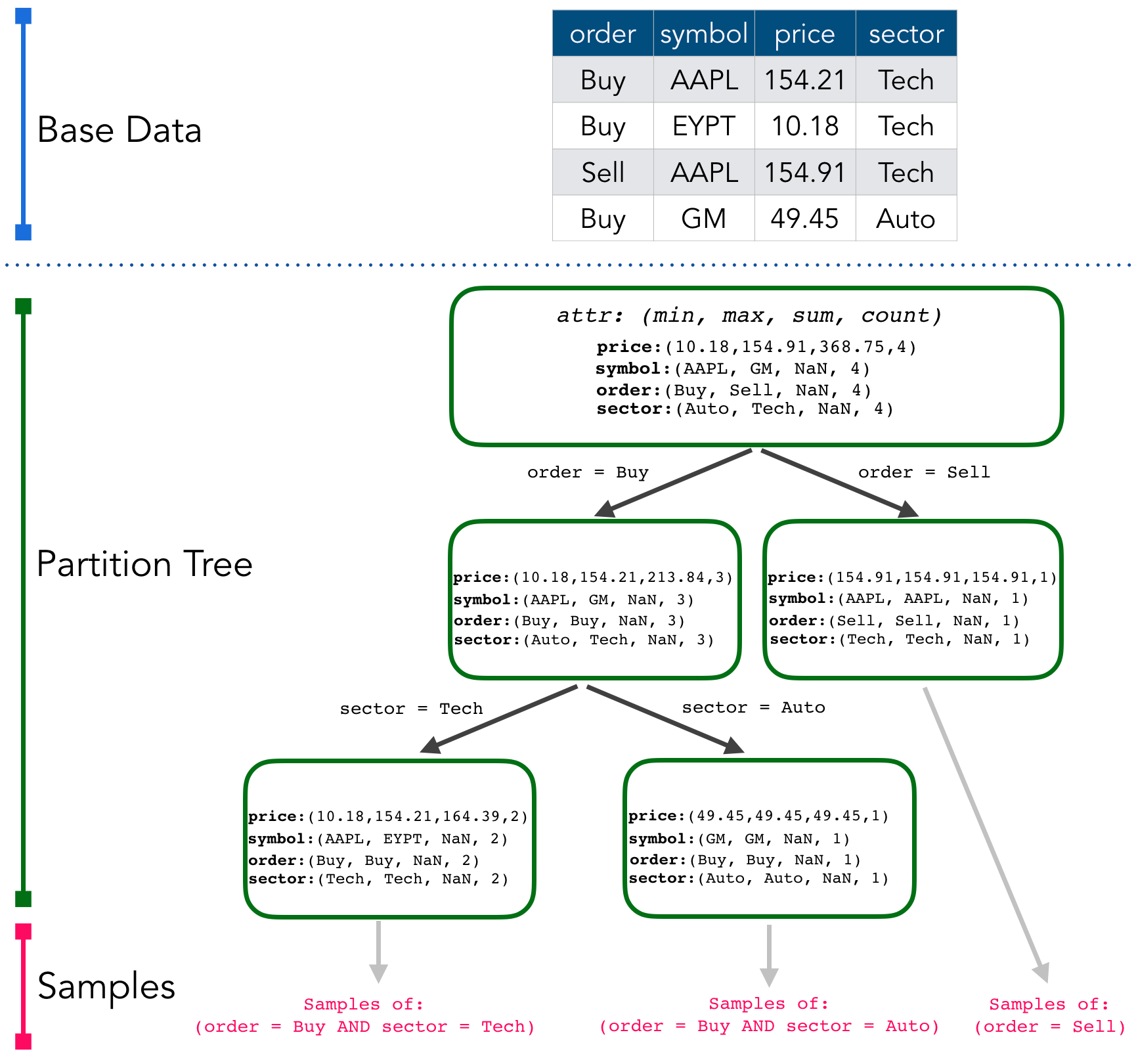}
    \caption{The core data structure in \sys is based on the PASS data structure~\cite{liang2021combining} that summarizes a dataset with a tree of aggregates at different levels of resolution (granularity of partitioning). Associated with the leaf nodes are stratified samples. The two stage synopsis structure can be optimally partitioned to minimize error.  \label{fig:teaser}}
\end{figure}

There are a few notable differences from the ``streaming'' setting. First, most data streaming models do not support arbitrary record deletion, i.e., as studied in~\cite{poepselland}. We find that in many use-cases limited support for deletion is needed due to records that are invalidated through an out-of-band, asynchronous data process like fraud detection or financial auditing. Next, most streaming settings enforce a single pass over the data with limited overall memory. We do not make this assumption and allow for archival storage and slow access to old data. This is a more realistic AQP setting where all data are stored, however, there is limited working memory for a fast, approximate query answering service. 

\subsection{Related work}
There is significant research in histograms and their variants that is highly relevant to this project~\cite{jagadish1998optimal, koudas2000optimal, jagadish2001global}.
V-Optimal histograms construct buckets to minimize the cumulative variance~\cite{jagadish1998optimal}. There are works on
multi-dimensional histograms~\cite{lazaridis2001progressive}, and histograms on the streaming/dynamic setting~\cite{guha2006approximation, gilbert2002fast}.
Like histograms, \sys constructs partitions over attribute domains and aggregates within the partition. However, we contribute different partition optimization criteria than typically used in histograms and novel techniques based on geometric data structures to scale partitioning into higher dimensions. Furthermore, our system works in the general dynamic setting, unlike~\cite{gilbert2002fast} where the number of total items must remain the same.
Another related area of research is into mergeable summaries that compute a partition of the data and optimize sampling at a data partition level~\cite{rong2020approximate, liang2020fast, agarwal2012mergeable, gan2020coopstore, poepselland}. \reviewtwo{The \syn used in \sys very much behaves like a mergeable summary but it combines low-error statistics over pre-computed partitioning of the space, along with stratified sampling to improve the error over comparing to traditional mergeable summaries for aggregation queries.}
Furthermore, some prior work mostly focuses on a streaming setting without support for deletion~\cite{poepselland}.
Similarly, sketches~\cite{cormode2011sketch, cormode2011synopses} have been used to find a summary of data to answer approximately a variety of queries efficiently. However, they also do not handle arbitrary range queries using space independent of the size of the full database. Mergeable summaries and sketches usually focus on optimizing different types of problems such that frequency queries, percentile queries, etc.
This paper shows how to operationalize a general DAQP system for aggregation queries with both systems and algorithmic contributions relating to the design of dynamic synopses and their continuous optimization. Our system can handle arbitrary updates and can estimate any arbitrary predicate query with provable confidence intervals.

In databases, a number of tree-based indexes, such as the improved $R^*$ tree~\cite{jurgens1998r}, have been used to support range aggregation queries efficiently.
The space of such indexes is (super-)linear with respect to the input items.
The query procedures need to have access to the entire tree-index that contains the entire dataset. That leads to high communication throughput or high I/O operations comparing to \sys where the queries are executed in a small synopsis of data stored in a local machine or RAM with zero communication throughput.
In another line of work, tree-based data structured are used to return a set of $k$ uniform samples in a query range. More specifically, in~\cite{joshi2008materialized, wang2015spatial} the authors construct indexes such that given a query range $Q$ and a parameter $k$, they return $k$ uniform samples from the input items that lie inside $Q$. These samples can be used to estimate any aggregation query in the range query $Q$.
There are several issues with these indexes in our setting.
First, the design of the index in~\cite{joshi2008materialized} makes their structure inherently static and it cannot be maintained efficiently.
Furthermore, the estimation error in both indexes is the same as the error in the simple uniform random sampling schema. In Section~\ref{sec:experiments}, we show that the error of our new index in real data sets is always less than half of the error in uniform random sampling, so our new index always outperforms these range sampling indexes.
In addition, the communication throughput or the I/O operations during a query procedure of these indexes is a function of $N$, i.e. the size of the input set, so they cannot be used on big data.
Finally, the dynamic tree structure in~\cite{lazaridis2001progressive} can store a synopsis of data in a tree-based index and use only this synopsis/index to return estimations of a range aggregation queries, which is also the case in our system \sys. However, there are two main differences. The index in~\cite{lazaridis2001progressive} returns a good estimation only if an assumption about the spatial/value-domain distribution of the data is made, while \sys uses stratified sampling and it always returns unbiased estimators with small error without assuming any distribution over the data. Furthermore, while their partition tree in~\cite{lazaridis2001progressive} can handle dynamic updates, its structure/partitioning remains unchanged. In our index we maintain a near-optimal partitioning over the updates. As we show in Section~\ref{sec:exp-del}, running experiments on real data, re-partitioning is essential in order to maintain a small error.

Dynamic AQP problems have been discussed in prior work~\cite{garofalakis2001approximate}, however, most existing systems have focused on a static data warehousing setting~\cite{agarwal2013blinkdb}. The Aqua system~\cite{acharya1999aqua} did consider the maintenance of its synopsis data structures under updates. However, these synopses were relatively simple and only samples and histograms. Furthermore, we discuss systems issues such as catch-up processing that was not discussed in~\cite{acharya1999aqua} or any subsequent work~\cite{gibbons2002fast}.  

Many new AQP techniques use machine learning.
The basic ideas exist for a while, e.g., ~\cite{jermaine2003robust,jin2006new}. Recently, there are more comprehensive solutions that train from a past query workload~\cite{park2017database} or directly build a probabilistic model of the entire database~\cite{hilprecht2019deepdb, yang2019deep}.
We show that these systems are not optimized for a dynamic setting. Even when they can be updated efficiently with warm-start training, their throughput is much lower than \sys.

\subsection{Partition Trees for AQP}
\label{sec:partTreeAQP}
We propose a new dynamic data synopsis and optimization strategy that is an extension of our previous work~\cite{liang2021combining}. In particular, we proposed a system called PASS (which we call \prior for ``static partition tree''). 
\prior synopses are related to works such as~\cite{lazaridis2001progressive} in the data cube literature and hybrid AQP techniques~\cite{peng2018aqp++}.
We showed that with appropriate optimization of the partitioning conditions, an \prior could achieve state-of-the-art accuracy in AQP problems.

\subsubsection{Construction}
An \prior is a synopsis data structure used for answering aggregate queries over relational data.
To use \prior, the user defines an \emph{aggregation column} (numerical attribute to aggregate) and a \emph{set of predicate columns} (columns over which filters will be applied).
An \prior consists of two pieces: (1) a hierarchical aggregation of a dataset, and (2) a uniform sample of data for each of the leaf partitions (effectively a stratified sample over the leaves).
The system returns a synopsis that can answer SUM, COUNT, AVG, MIN, and MAX aggregates over the aggregation column filtered by the predicate columns.
Figure \ref{fig:teaser} illustrates a partition tree synopsis over toy stock-order data.

To understand how this structure is useful, let us overview some of its formal properties.
A \emph{partition} of a dataset $\fullData$ is a decomposition of $\fullData$ into disjoint parts $\fullData_1,...,\fullData_B$.
Each $\fullData_i$ has an associated partitioning condition, a predicate that when applied to the full dataset as a filter retrieves the full partition.
Partitions naturally form a hierarchy and can be further subdivided into even more partitions, which can then be subdivided further.
A \emph{static partition tree} $\PASStree$ is a tree with $B$ nodes (where each node corresponds to a partition) with the following invariants: (1) every child is a subset of its parent, (2) all siblings are disjoint, and (3) the union of all siblings equals the parent.

In an \prior, each node of the tree is associated with SUM, COUNT, MIN, and MAX statistics over the items in $\fullData$ that lie inside the node.
\prior synopses have a flexible height to tradeoff accuracy v.s. storage.
In shorter trees, the leaf nodes of an \prior can cover large subsets of data and vice versa in deeper trees.
Note how each layer of the tree in Figure \ref{fig:teaser} aggregates the lower layer over coarser-and-coarser aggregation conditions (first by ``sector'' and then by ``order type'').

This structure works well when the queries align with partition boundaries. For example, a user aggregating total orders by ``order type'' in Figure \ref{fig:teaser} would get an exact answer with no approximation. The challenge is to answer queries with predicates that partially intersect partitions.
Due to the tree invariants, the set of partial intersections can be fully determined at the leaf nodes.
To estimate the contributions of these partial intersections, an \prior associates a uniform sample of tuples \emph{within that partition} for each leaf node.

\subsubsection{Query Processing}
Using an \prior, a user can estimate the result of a query as follows. Essentially, the query processing algorithm identifies ``fully covered'' nodes that are contained in the query predicate and ``partially covered'' ones that overlap in some way. Exact statistics from the ``fully covered'' nodes can be used, while estimates can be used to determine the contribution of ``partially covered'' ones. We present SUM, COUNT, AVG for brevity, but it is also possible to get estimations for MIN and MAX.

\vspace{0.5em} \noindent  \textbf{Step 1: Frontier Lookup. }  Given a query predicate $q$, traverse the tree top-down to retrieve two sets of nodes partitions: $R_{cover}$ (nodes that fully cover the predicate) and $R_{partial}$ (nodes that partially intersect the predicate). Nodes that do not intersect the predicate can be ignored.

\vspace{0.5em} \noindent \textbf{Step 2: Partial Aggregation } For each partition in $R_{cover}$, we can compute an exact ``partial aggregate'' for the tuples in those partitions. For a SUM/COUNT query $q$: $agg =  \sum_{R_i \in R_{cover}} SUM(R_i)$,
for an AVG query, we weight the average by the relative size of the partition: $agg =  \sum_{R_i \in R_{cover}} SUM(R_i) \frac{N_i}{N_q}$, where $N_i$ is the size of the partition $R_i$, $N_q$ is the total size in all relevant partitions of query $q$, and $SUM(R_i)=\sum_{t\in R_i\cap \fullData}t.a$ is the sum of the aggregation values of all tuples in the partition $R_i$.

 \vspace{0.5em} \noindent \textbf{Step 3: Sample Estimation. } Each partition in $R_{partial}$ is a leaf node with an associated stratified sample. Within each stratified sample, we use standard AQP techniques to estimate that partition's contribution to the final query result~\cite{agarwal2013blinkdb}. For completeness, we include those calculations here. Suppose a partition $R_i$ has a set $\samples_i$ of $m_i$ samples and there are $N_i$ total tuples in $R_i$.
 We can formulate COUNT, SUM, AVG as calculating an average over transformed attributes:
$f(\samples_i)=\frac{1}{m_i}\sum_{t\in \samples_i}\phi_q(t)$,
where $\phi_q(\cdot)$ expresses all the necessary scaling to translate the samples in query $q$ into an average query population. Let $Predicate(t, q)=1$ if tuple $t$ satisfies the predicate of query $q$, and $0$ otherwise, we have
\begin{itemize}
    \item COUNT: $\phi_q(t)=Predicate(t,q)\cdot N_i$
    \item SUM: $\phi_q(t)=Predicate(t,q)\cdot N_i\cdot t.a$
    \item AVG: $\phi_q(t)=Predicate(t,q)\cdot \frac{m_i}{\sum_{t\in \samples_i}Predicate(t,q)}\cdot t.a$
\end{itemize}
We run such a calculation for each partition that is partially covered. These results are combined with a weighted combination like before. For SUM/COUNT queries it is: $samp = \sum_{R_i \in R_{partial}} f(S_i)$. And for AVG queries, it is: $samp = \sum_{R_i \in R_{partial}} f(S_i) \cdot \frac{N_i}{N_q}$. $N_i$ and $N_q$ can be exactly retrieved from the statistics computed for each partition.

\vspace{0.5em} \noindent \textbf{Step 4: Final Estimate. } The results can be found by taking a sum of the two parts: $result = samp + agg.$ For this result estimate, confidence intervals can be calculated using standard stratified sampling formulas.

\vspace{0.25em} \noindent
\textbf{PASS and \sys comparison. }
As we noted, \sys is an extension of PASS in the dynamic setting.  The main differences and novelties of our new system \sys comparing to PASS are the following: i) PASS finds a static partitioning that is not changing after insertions and deletions of items. In \sys we propose algorithms (Subsection 5.4) that automatically check if a re-partitioning is needed after the dynamic updates.  ii) Even if re-partitioning is allowed in PASS, the algorithms we proposed in ~\cite{liang2021combining} do not run efficiently in the dynamic setting. Here we propose dynamic indexes and algorithms with theoretical guarantees that perform much faster than the algorithms in PASS. iii) Even if we use our new dynamic algorithms in PASS, there is no mechanism to compute the exact statistics of the nodes after a re-partitioning happening and there is no mechanism handling the updates as the re-partitioning is executed. \sys can improve the estimators in the nodes of \syn after a re-partitioning, while handling new dynamic updates and new queries. iv) Last but not least, we implement \sys on Apache Kafka, so it can be used by real database systems.
\section{System Architecture}
\label{sec:implementation}
In this section, we describe the \sys architecture.
\vspace{-0.5em}
\subsection{Construction and Optimization API}
\label{subsection:constructAPI}
\reviewtwo{
First, we overview how users construct synopsis data structures in \sys.
In order to simplify the description of our system, for now, we consider that
unlike systems like BlinkDB~\cite{agarwal2013blinkdb}, \sys does not use a single synopsis to answer all queries. Instead, \sys constructs a different synopsis for each different query template that the user is interested in. In Subsection~\ref{subsection:discussion} we describe how we can construct a single synopsis that can answer different query templates.
Much like index construction in a database, users choose which attributes to include in the synopsis structure. Each synopsis can answer query templates of the following form:}
\begin{verbatim}
SELECT SUM/COUNT/AVG/MIN/MAX(A) FROM D
WHERE Rectangle(D.c1,...,D.cd)
\end{verbatim}
\reviewtwo{
where $A$ is an aggregation attribute and $c_1,...,c_d$ are predicate attributes used in some rectangular predicate region (a conjunction of $>,<,=$ clauses). The \textbf{dimensionality} of a synopsis is the number of predicate attributes $d$.
To construct a synopsis, the user must define the following basic inputs:
\vspace{-0.5em}
\begin{itemize}
    \item \textbf{Aggregation Attribute and Aggregation Function. } An attribute $A$ that is the primary metric for aggregation. 
    \item \textbf{Predicate Attributes. } A collection of $d$ columns $c_1,...,c_d$ that are used to filter the data prior to aggregation. 
    \item \textbf{Memory Constraint. } The maximum amount of space that the synopsis can take.
\end{itemize}
The user can also optionally define a query processing constraint, i.e, the maximum bytes of data that the system should process in answering a query. 
Finally, the user has the option to set a historical data limit which is how much historical data to include in the synopsis, i.e., the earliest time-step of data included in the system.
In Subsection~\ref{subsection:discussion} we also show how \sys uses these constraints to generate a synopsis with low error.
}

Beyond these basic knobs that are relevant to most AQP systems, there are two other considerations discussed in this paper:  

\vspace{0.25em} \noindent \textbf{Catch-Up Processing. } Constructing a synopsis will require some amount of computational time. While incremental maintenance might be efficient, constructing the initial synopsis $S^{(0)}$ from the initial database state $\fullData^{(0)}$ might be very expensive if there is a significant amount of initial data. As the initial $S^{(0)}$ is being constructed new data will arrive, and the system will require additional processing to catch up. \sys optimizes the catch-up process using a multi-threaded system and approximate internal statistics for the partition tree. This process minimizes the amount of time where the system is unable to process new data or queries. The user decides how much processing to expend during catch up, the quicker the system is ready, the higher the error will be.

\vspace{0.25em} \noindent \textbf{Throughput. } The maximum data throughput is the maximum rate of insertions and deletions that the system can support. Throughput depends on the complexity of the synopsis used.

\subsection{Data and Query API}
For processing queries and data, we adopt the PSoup architecture where both queries and data are streams~\cite{chandrasekaran2003psoup}.
\sys supports three types of requests: insertion of a new tuple, deletion of an existing tuple and querying of the existing tuples. Thus, there are three Kafka topics \textsf{insert(tuple)}, \textsf{delete(tuple)}, and \textsf{execute(query)}.

The use of Kafka, with its timing and delivery guarantees, simplifies the query processing semantics.
The system will process the incoming stream of queries in order. Each query will have an arrival time $i$, which is the current database state at the time at which the query is issued. Therefore, we define $q^{(i)}_j$ as the $j$th query in the sequence that arrives at database state $i$. Query results should reflect all of the data that has arrived until the time point $i$.

\section{Dynamic Partition Trees}\label{sec:dpt}
We discuss how Dynamic Partition Trees (\syn) are constructed, how they answer queries, and how they are maintained under updates. Structurally, a \syn is essentially the same data structure as an \prior; however, the way that the partition statistics and samples are represented differ to allow for incremental maintenance. Figure \ref{fig:dpt} summarizes the basic update process. 

\begin{table}[t]
\footnotesize\vspace{-0.5em}
\caption{Table of basic notation}
\centering
 \begin{tabular}{|c|c|c|c|} 
 \hline
 $\fullData$ & Full database & $H_i$ & $H\cap R_i$ \\ \hline
 $N$ & $|\fullData|$ & $m_i$ & $|\samples_i|$ \\ \hline
 $\samples$ & Set of reservoir samples & $h_i$ & $|H_i|$ \\ \hline
 $H$ & Set of catch-up samples & $m$ & $|\samples|$ \\ \hline
 $R_i$ & Partition/bucket/rectangle & $t$ & Tuple in $\fullData$ \\ \hline
  $|R_i|$ & $|R_i\cap \samples|$ & $t.a$ & Aggregation value of tuple $t$ \\ \hline
 $N_i$ & $\fullData\cap R_i$ & 
 $\PASStree$ & Partition tree in \syn \\ \hline
 $\samples_i$ & $\samples\cap R_i$& &\\\hline
\end{tabular}
\label{Table:Notation}
\end{table}

\begin{figure}[t]
    \centering
    \includegraphics[scale=0.25]{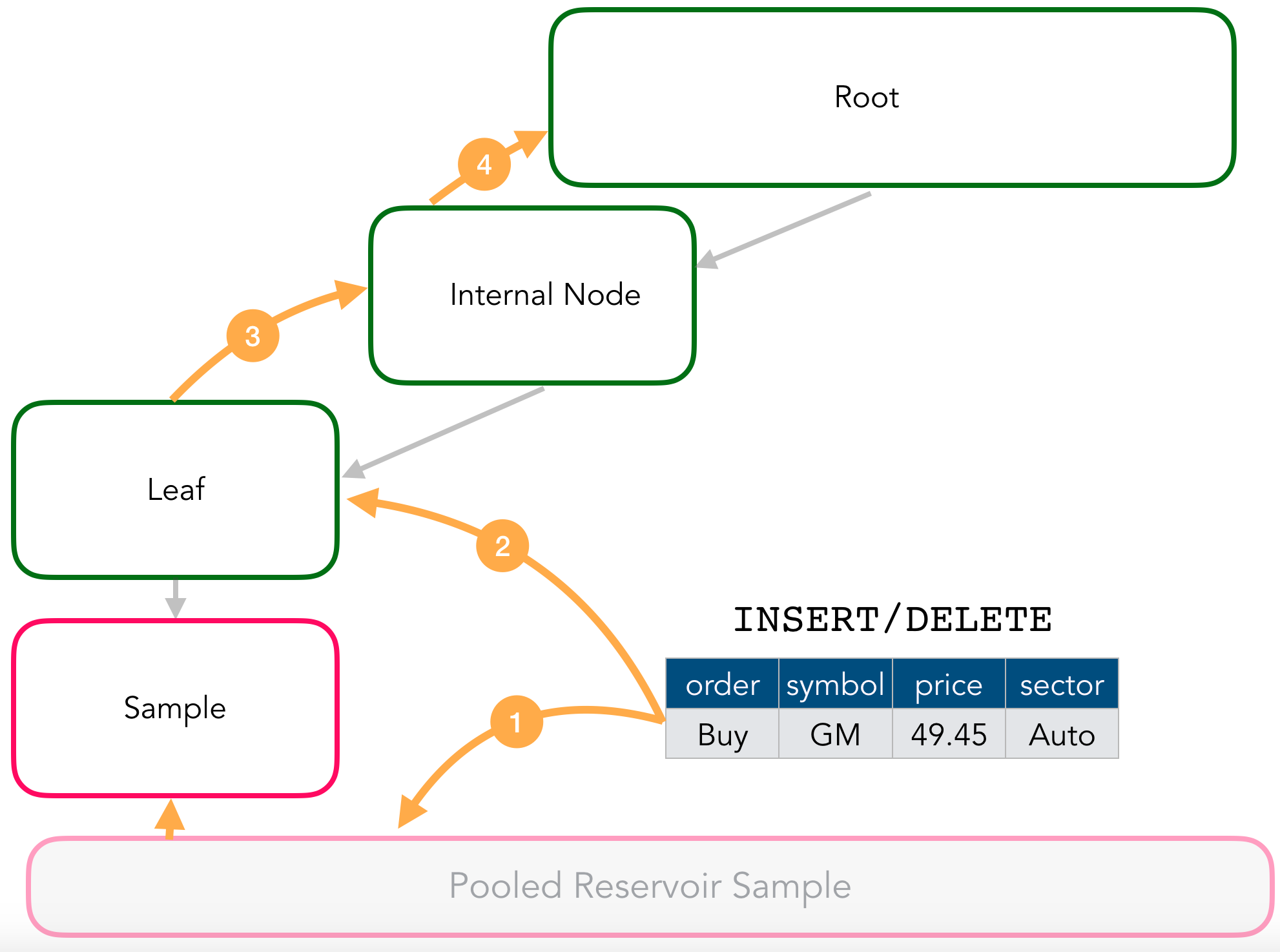}
    \caption{The \syn update process for an insertion or deletion. (1) A set of samples is maintained using a reservoir sampling algorithm. (2) The leaf node statistics are incrementally updated. (3) The updated statistics from the leaf node propagate to the parents. (4) Updated statistics from the parents propagate all the way to the root.  \label{fig:dpt}}
\end{figure}

\subsection{Incrementally Maintaining Nodes}
\label{sec:IncrMaintNodes}
Each node defines a partition and contains statistics (the SUM, COUNT, MIN, and MAX aggregates) of the data contained in that partition.
The key challenge is to keep these statistics up-to-date in the presence of insertions and deletions. 
When an insertion or deletion arrives, an entire path of nodes from the leaf to the root will have to be updated.

\vspace{0.5em} \noindent \textbf{\syn Nodes: } First, we discuss how we represent the statistics in a \syn node. Since the SUM and COUNT are easy to incrementally maintain under both insertions and deletions, we simply store a single SUM and COUNT value for each aggregation attribute. The MIN and MAX values are harder to incrementally maintain. To store the MIN and MAX values, we store the top-k and the bottom-k values in a MIN/MAX heap respectively. The top value of these heaps is equal to the MIN and MAX of all the data in the node. 

\vspace{0.5em} \noindent \textbf{Insert New Record: } When a new record is inserted,  we start from the root and we find the leaf node that the record is contained in. Then, we increment the SUM and COUNT statistics accordingly. Finally for MAX/MIN, we push the new aggregation values onto the heap. If the heap exceeds the size limit $k$, then the bottom value on the heap is removed. 

\vspace{0.5em} \noindent \textbf{Delete Existing Record: } When an existing record is deleted, we find the leaf node that the record is contained in. Then we decrement the SUM and COUNT statistics accordingly. Finally for MAX/MIN, if that aggregation value is contained in the heap it is removed from the heap. \reviewone{This might make the heap smaller than $k$.  Repeated deletes from the same node might fully empty the heap. We stop removing values from the heap when there is only one value left. When the heap reaches a single element the MIN/MAX estimates received from the nodes are outer approximations where the estimated value is larger than the MAX and smaller than the MIN.}
\vspace{-1em}
\subsection{Maintaining Stratified Samples}
Next, we describe how to maintain the samples associated with leaf nodes. \reviewtwo{We maintain samples using the \emph{proportional allocation} of stratified sampling. The number of samples in each stratum should be proportional to the size of the stratum over the size of the population.
We use a modified version of the well-known technique of \emph{reservoir-sampling}~\cite{vitter1985random} under updates~\cite{gibbons2002fast}.
The details of how we implement this are interesting.
Conceptually, each leaf node is associated with a physically disjoint sample of just that partition, i.e., a stratified sample.
Instead of implementing physical strata for the stratified sampling, we implement large enough virtual partitions of a single global sample. 
This global sample can be maintained using a reservoir sampling algorithm and makes it easier to control the overall size of the synopsis under insertions/deletions as well as simplifies concurrency control}.

\vspace{0.5em} \noindent \textbf{Sample Representation: } The \syn maintains a ``pooled'' sample (all the relevant samples in a single data structure). This set of samples has a target size of $2m$ tuples. At the construction time, we choose a set $\samples$ of $2m$ uniform random samples from $\fullData$. The update procedure ensures that there are always between $m\leq |\samples|\leq 2m$ samples. The leaf nodes index into this ``pooled'' sample selecting only the relevant data to their corresponding partitions. \reviewtwo{Since we get a set of uniform samples, if we make sure that our strata are large enough, the number of samples we get in each stratum satisfies the proportional allocation requirements, up to a factor of $2$, with high probability. We formally show it in Appendix~\ref{appndx:proportional}. This is also the case in the next update operations.}

\vspace{0.5em} \noindent \textbf{Insert New Record: } Suppose we insert a new tuple $t$. If $|\samples|<2m$ we add $t$ in $\samples$. If $|\samples|=2m$, we choose $t$ with probability $\frac{|\samples|}{|\fullData|}$. \reviewone{If it is selected then we replace a point from $\samples$, chosen uniformly at random, with $t$.} 

\vspace{0.5em} \noindent \textbf{Delete Existing Record: } Next, suppose that we delete a tuple $t$ from $\fullData$. If $t\notin \samples$ we do not do anything. If $t\in \samples$ then we check the cardinality of $\samples$. If $|\samples|>m$ then we only remove $t$ from $\samples$. If $|\samples|=m$ then we skip the set $\samples$ and we re-sample $2m$ items from $\fullData$. As shown in~\cite{gibbons2002fast} this procedure always maintain a set of uniform random samples.
\reviewtwo{As shown in Section~\ref{sec:re-partition} we always check if a stratum is under-represented after a number of deletions. If this is the case then we need to re-sample and re-partition.}
Using a simple dynamic search binary tree of space $O(m)$ we can update the samples $\samples$ stored in $\PASStree$ in $O(\text{height}(\PASStree))$ time.

\begin{figure}[t]
    \centering
    \includegraphics[scale=0.25]{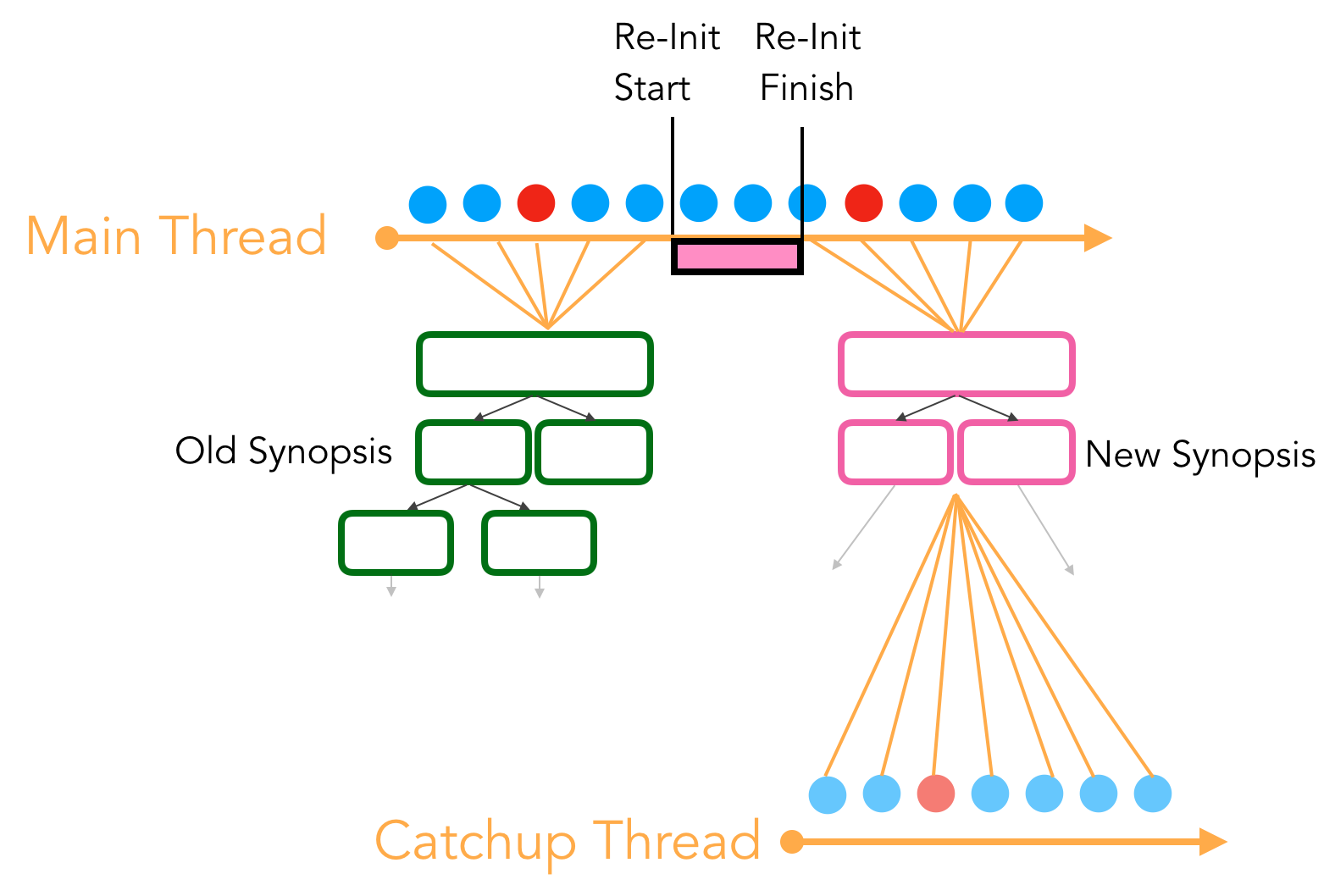}
    \caption{\sys synopses can be re-initialized online using a multi-threaded implementation to minimize unavailability\label{fig:catchup}}
\end{figure}

\subsection{Re-initialization and Catch-Up}\label{sec:catchup}
As we noted before, repeated deletes on the same leaf partition can degrade the accuracy of the synopsis. As we will see in the next section, it is also possible for repeated insertions to degrade the accuracy as well. In such cases, re-initialization of the \syn may be needed where the data structure is re-built and re-optimized over existing data. 

Enabling periodic re-initialization is crucial for reliable long-term deployment but is challenging because new data will not simply stop arriving during the re-initialization period.
We employ a multi-threaded approach to minimize any period unavailability for processing new data arrival as well as new queries (Figure \ref{fig:catchup}).  When re-initialization is triggered, the main processing thread initiates the construction of a new \syn synopsis
and the following steps are performed:
\vspace{-0.5em}
\begin{enumerate}
	\item Optimization Phase (In Parallel)
	\begin{itemize}
\item The partition optimization algorithm analyzes the data in the pooled reservoir sample to determine the optimal new partitioning criteria. It returns a new empty \syn with no node statistics.
	\item In parallel with (Step 1), the old synopsis is maintained under all insertions and deletions that happen during the optimization algorithm. Queries can still be answered with the old synopsis.
	\end{itemize}
    \item  (Blocking) Approximate node statistics are populated into the new synopsis using the pooled reservoir sample $S$ (note, that this will reflect any data that arrived during the optimization phase). This is the only blocking step in the re-initialization routine and new data and queries will have to wait until completion.
    \item The old synopsis is discarded.
	\item The system resamples a uniform sample of data from archival storage to be the new pooled reservoir sample. Queries and results can still be processed on the new synopsis even without a sample.
	\item Random samples of historical data are used to improve the node statistics in the background until a user-specified ``catch-up'' time. 
\end{enumerate}

This process is the key difference between an \prior and an \syn, where after catch-up the node statistics may be inexact.
However, this old data is propagated in a random order, which means that the SUM,COUNT,AVG values in each node will be unbiased estimates of their full data statistics.  The duration of the catch-up phase can be chosen by the user. For example, in our experiments, the catch-up phase does not stop until we get $0.1\cdot|\fullData|$ samples.
It is worth noting that queries close to the beginning of the catch-up phase will have a higher error, however queries towards the middle or the end of the catch-up phase will have a smaller error.
In Section~\ref{sec:re-partition}, we describe how to trigger re-initialization.
Furthermore, there is only one step (2) where the synopsis is unavailable to process queries and data (has the duration of 100s of milliseconds in our experiments).


\subsection{Answering Queries With a \syn}
\label{sec:estimators}
The query procedure does not access the entire data so the communication throughput (or the I/O operations) is zero.
The basic structure of the result estimator is the same as before, especially for the $R_{partial}$ partitions. However, there are a few key changes due to the nature of the catch-up phase. In \prior, for each partition in $R_{cover}$, we can compute an exact ``partial aggregate'' for the tuples in those partitions and combine the partial aggregates. In a \syn, this process changes
considering the estimations we get from the catch-up samples.
Overall, the estimation of a partition $R_i\in R_{cover}$ consists of i) estimation using the catch-up samples $H$ and the formulas of Section~\ref{sec:partTreeAQP}, ii) the exact statistics of the new inserted tuples in $R_i$, and iii) the exact statistics of the deleted tuples in $R_i$ (recall that the quantities in ii), iii) are stored and maintained as described in the Incrementally Maintaining Statistics in Section~\ref{sec:IncrMaintNodes}).
By taking the sum of i), ii) and subtracting iii) we get the unbiased estimation in partition $R_i$.


Let $H$ be the set of catch-up samples, $H_i\subseteq H$ be the subset of $H$ in partition $R_i$, and $h_i=|H_i|$. All basic notations are defined in Table~\ref{Table:Notation}.
The formulas for estimating COUNT and SUM queries in both $R_{cover}$, $R_{partial}$ from Section~\ref{sec:partTreeAQP} contain the factor $\frac{N_i}{m_i}$ or $\frac{N_i}{h_i}$, while the formulas for estimating the AVG contain the factor $\frac{N_i}{N_q}$. In \syn we do not have the exact values for $N_i$. Instead, we use an estimate of the size of the partition $R_i$ denoted by $\hat{N_i}$. In particular we use the catch-up samples $H$ to estimate $\hat{N_i}=\frac{h_i}{h}N$.

\subsubsection{Confidence Intervals}
While the estimators do not significantly change from an \prior to a \syn, the confidence intervals are calculated very differently. This is because there are now two sources of errors: estimation errors due to the stratified samples and estimation errors in the node statistics. Both these sources of errors have to be integrated into a single measure of uncertainty.
Assuming that all partitions are large enough, the central limit theorem can be used to asymptotically bound the estimation error for SUM/COUNT/AVG queries. Informally, the central limit theorem states that this asymptomatic error is proportional to the square-root of the ratio of estimate variance and the amount of samples used $\propto \sqrt{\frac{var(est_i)}{m_i}}$.
We simply have to match terms to this formula for all sample estimates and all node estimates because both are derived from samples.

\mparagraph{Error in Node Estimates} First, let's account for all the uncertainty due to catch-up. Recall that $H$ is the set of catch-up samples we have considered so far and $H_i\subseteq H$ is the samples in partition $R_i$ with $h_i=|H_i|$. We note that we do not store the set $H$ or the subsets $H_i$, instead we only use the new catch-up samples to continuously improve the statistics we store in the nodes. Using the notation in the previous section, we can calculate the catch-up variance $\nu_c$:
\[
\nu_c(q) = \sum_{R_i\in R_{cover}} w_i^2\frac{var(\phi_q(H_i))}{h_i}
\]
where $w_i= \frac{\hat{N_i}}{\hat{N_q}}$ for AVG queries and $w_i=1$ for SUM/COUNT queries. Calculating $\phi_q(H_i)$ is straight-forward. We simply store additional information that allows us to efficiently calculate the variance. For any node $i$ of $\PASStree$ we store $h_i$, $\sum_{t\in H_i}t.a^2$, $\sum_{t\in H_i}t.a$.

\mparagraph{Error in Sample Estimates}
For a partition $R_i \in R_{partial}$, let $\samples_i\subseteq \fullData$ be the set of samples in $\samples$ that lie in partition $R_i$ and let $m_i=|\samples_i|$.
Like the catch-up variance, we can calculate the sample estimate variance $\nu_s$:
$$\nu_s(q) =  \sum_{R_i\in R_{partial}}w_i^2\frac{var(\phi_q(\samples_i))}{m_i}$$
We can calculate an overall confidence interval as:
$$\pm z \cdot \sqrt{\nu_c(q)  + \nu_s(q) } $$
where $z$ is a normal scaling factor corresponding to the desired confidence level, e.g., $z=1.96$ for 95\%. As before, $w_i= \frac{\hat{N_i}}{\hat{N_q}}$ for AVG queries and $w_i=1$ for SUM/COUNT queries.
In Appendix~\ref{appndx:varest}, we show analytically all formulas for computing the variance under different types of queries.

\section{Optimal \syn Partitioning}
\label{sec:partAlg}
\newcommand{\poly}{\textrm{poly}}
We next describe a new dynamic partitioning algorithm designed for the dynamic setting.

\subsection{Preliminaries and Problem Setup}
\label{subsection:prelimandsetup}
The partitioning algorithm analyzes the pooled reservoir sample of data to determine how best to partition the dataset.
The goal of the partitioning algorithm is to find a partitioning such that the subsequent queries issued to the \syn have low-error. 
Surprisingly enough, the partitioning algorithm does not need an exact query workload to perform this optimization. 
It simply needs a focus aggregation function (e.g., SUM, COUNT, AVG) and finds a partitioning that minimizes the worst-case query error for sufficiently large predicates.
Given a set of $O(m)$ samples $S$, the goal is to construct a data structure that supports the following operations.
(i) Insert or delete a sample from $\samples$ efficiently, and (ii) when a partitioning request comes, it creates a near-optimum partition tree $\PASStree$ in $o(m)$ time.

Let $Q$ be a set of possible aggregate queries with a predicate. \reviewtwo{And, let $\Theta$ be the set of all \syn synopses with $k$ leaf nodes. Each $\syn$ consists of a hierarchical rectangular partitioning having $k$ partitions/buckets in the last level.} The main optimization we need to solve to find a good enough partitioning is to minimize the maximum error over the query workload:
\vspace{-0.8em}
\begin{equation}
    \min_{\PASStree \in \Theta} \max_{q \in Q} \textsf{Error}(q,\PASStree)
\end{equation}
The error is defined as the length of the confidence interval, as defined in the previous section. Since the catch-up variance is usually extremely smaller than the sample estimate variance, we focus on minimizing the maximum length of the confidence interval with respect to the sample estimate variance $\nu_s(\cdot)$.
\reviewone{Indeed, the catch-up phase is running in the background retrieving a large number of samples as the time passes. Hence, the estimations in every node of the \syn are constantly improving reducing $var(\phi_q(H_i))$ for each partition $R_i$. On the other hand, the number of samples $S$ that we are storing explicitly are limited making the sample estimate variance larger than the catch-up variance.}
For simplicity, when we say variance we always mean the sample estimate variance.

Instead of looking over all possible queries to minimize the maximum error, one only needs to focus on single partitions to ensure they do not have ``high-variance'' sub-partitions (see Appendix~\ref{appndx:partAlg}).
Indeed by considering only these sub-partitions we can still get a $\sqrt{k}$-approximation for COUNT and SUM queries over the optimum partition considering all queries (for $k$ leaves). The approximation factor improves to $\sqrt{2}$ for $d=1$. For AVG queries the error of the optimum partition of this simplification is the same with the maximum error considering every possible query. 
The error of a query $q$ inside a leaf node (partition) $R_i$ is defined (expanding the equations from the definition of $\nu_s(\cdot)$) as
$$\frac{N_i^2}{m_i^3}\!\left[m_i\!\sum_{t\in q}\!t.a^2\!-\!\left(\!\sum_{t\in q}\!t.a\!\right)^2\right], \frac{m_i\!\sum_{t\in q}\!t.a^2\!-\!\left(\!\sum_{t\in q}\!t.a\!\right)^2}{m_i|q\cap \samples|^2}.$$
for SUM/COUNT and AVG queries, respectively.
Thus, the optimization problem reduces to finding partitions that do not contain a high-variance ``rectangle'' of data.



Given a rectangle $R$, the goal is to find a rectangular query within $R$ with maximum variance among all possible queries in $R\cap S$. For now, we assume that we have a dynamic index $\maxVar$ with near-linear space such that given a query rectangle $R$, it returns a query $q$ within $R$ with $\nu_s(q)\geq \frac{1}{\gamma}\optVar(R)$, \reviewone{for an approximation factor $\gamma>1$}, in $O(M)$ time, where $\optVar(R)$ is the variance of the maximum variance rectangular query in $R$. Let $\maxVar(R)$ be the variance of the query returned by the index $\maxVar$. We describe this index with more details in Subsection~\ref{sec:partHighDim}.


\subsection{Partitioning for $d=1$}
\label{subsection:1dpartitioning}
Now, we discuss how to solve the partitioning optimization problem in one dimension.
We present results for SUM and AVG queries. COUNT can be thought of as a special case of SUM with binary data.
The basic trick is to search over a discretized set of possible variance values.
For each value $e$, we try to construct a partitioning of $k$ partitions such that in each bucket the length of the longest confidence interval of a query is at most $e$.
By systematically reducing $e$ in each iteration, we control for the worst-case error.


\mparagraph{Bounding the Error}
The first step is to calculate the bounds for the maximum length of the largest possible confidence interval among queries that intersect one partition.
We assume that the aggregation value of any item in $\fullData$ is bounded by a maximum value $\mathcal{U}$ and a minimum non-zero value $\mathcal{L}$. We allow items to take zero values since this is often the case in real datasets but no item with positive value less than $\mathcal{L}$ or larger than $\mathcal{U}$ exists. We assume that $\mathcal{U}=O(\poly(N))$ and $\mathcal{L}=\Omega(1/\poly(N))$.
In Appendix~\ref{appndx:1dpartition}
we show that the length of the longest confidence interval is also bounded by $O(\poly(N))$ and $\Omega(1/\poly(N))$.

\mparagraph{Description of Algorithm}
We describe the partitioning algorithm for SUM queries. The procedure is identical for AVG queries.
For a parameter $\rho\in \mathbb{R}$ with $\rho>1$, let $E=\{\rho^t\mid t\in  \mathbb{Z}, \frac{\mathcal{L}}{\sqrt{2}}\leq \rho^t\leq N\mathcal{U}\}\cup\{0\}$, be the discretization of the 
range defined by
the lower and upper bound of the longest confidence interval (as defined in the previous paragraph).
We run a binary search on the values of $E$. For each value $e\in E$ we consider, we try to construct a partitioning of $k$ partitions such that in each partition the length of the longest confidence interval of a query is at most $e$. If there exists such a partitioning we continue the binary search with values $e'<e$. If there is no such a partitioning we continue the binary search with values $e'>e$. In the end, we return the last partitioning that we were able to compute.

It remains to describe how to check if a partitioning with $k$ buckets (intervals) with maximum length confidence interval at most $e$ exists.
A high level description of the algorithm is:
\begin{enumerate}
    \item For $i=1$ to $k$
    \begin{enumerate}
        \item Let $b_i$ be the $i$-th bucket with left endpoint
    $t_a$
        \item Binary search on samples $t_j$ to find the maximum bucket $b_i$ with error at most $e$
        \item If $\sqrt{\maxVar([t_a, t_j])}\leq e$
        \begin{enumerate}
            \item Continue search for values $>j$
            \item Else Continue search for values $<j$
        \end{enumerate}
    \end{enumerate}
    \item If the partitioning contains all samples construct $\PASStree$ using $b_i$ as its leaf nodes. Otherwise $\PASStree=\emptyset$.
\end{enumerate}
We start with the leftmost sample, say $t_1$, which is the left boundary of the first bucket. In order to find its right boundary we run a binary search on the samples $\samples$. Let $t_j$ be one of the right boundaries we check in the binary search, and let $b_1=[t_1, t_j]$. If $\sqrt{\MAXV(b_1)}\leq e$ then we continue the binary search with a sample at the right side of $t_j$ (larger bucket). Otherwise, we continue the binary search with a sample at the left side of $t_j$ (smaller bucket). When we find the maximal bucket with longest confidence interval at most $e$ we continue with the second bucket repeating the same process for at most $k$ buckets. In the end, if all samples in $\samples$ are contained in $k$ buckets then we return that there exists a partitioning (with $k$ buckets) with maximum variance at most $e$. If we cannot cover all samples in $k$ buckets then we return that there is no partitioning with $k$ buckets and maximum variance at most $e$.

\forfullversion{
\SetKwComment{Comment}{/* }{ */}
\begin{algorithm}
\caption{}\label{alg:one}
\KwData{$e, \samples, k, \eta$}
\KwResult{$\PASStree$}
$j \gets 0$\;
\For{$i=1$ to $k$}{
    $i_1\gets j+1$, $i_2\gets |\samples|-1$\;
    \uIf{$i_1>|\samples|$}{
        \textbf{break}\;
    }
    \While{$i_1\leq i_2$}{
        $i_3\gets \text{floor}((i_1+i_2)/2)$\;
        \uIf{$i_3-i_2+1\geq \eta$}{
            $W\gets \maxVar([t_a,t_{i_3}])$\;
            \uIf{$\sqrt{\frac{W}{i_3-i_2+1}}\leq e$}{
                $j\gets i_3$\;
                $i_1\gets i_3+1$\;
            }\Else{
                $i_2\gets i_3-1$\;
            }
        }\Else{
            $i_1\gets i_3+1$\;
        }
    }
    $b_i\gets [t_a, t_j]$\;
}
\uIf{$\left(\bigcup_i b_i\right)\cap \samples =\samples$}{
    Construct search binary tree $\PASStree$ using $b_i$ as leaf nodes\;
    \Return $\PASStree$\;
}\Else{
    $\PASStree\gets \emptyset$\;
}
\end{algorithm}
}


\textbf{Correctness. }
In Appendix~\ref{appndx:1dpartition} we use the monotonic property of the longest confidence interval (the bigger the bucket the larger the error) and we show $\sqrt{\optVar(b')}\leq \sqrt{\gamma \MAXV(b')}\leq \sqrt{\gamma} e'\leq \rho\sqrt{\gamma} \sqrt{\optVar(b^*)}$, where $b'$ is the bucket with the longest confidence interval in the returned partitioning, $e'$ is the smallest value in $E$ such that $\sqrt{\optVar(b^*)}\leq e'$, and $b^*$ is the bucket of optimum partitioning with the largest confidence interval.
For $d=1$ we have that $\gamma=4$ for SUM and AVG queries queries, so we get a partitioning where the maximum error is within $2\rho\sqrt{2}$ of the optimum error for SUM queries and within $2\rho$ of the optimum error for AVG queries.

\textbf{Running time. }
Since, $\mathcal{L}, \mathcal{U}$ are polynomially bounded on $N$ we have that $|E|=O(\log_\rho N)$ and it can be constructed in $O(\log_\rho N)$ time. The binary search over $E$ takes at most $O(\log \log_\rho N)$ steps.
We can decide if there exists a partitioning with error $e$ in $O(kM\log m)$ time. Overall, the running time of our algorithm is $O(kM\log m \log\log_\rho N)$. If $\rho$ is a constant, for example $\rho=2$, then the running time is $O(kM\log m \log\log N)$. In Appendix~\ref{appndx:partAlg}
we have that in $1$-dimension $M=O(\log m)$ for SUM and AVG queries. Notice that if we skip the $\log$ factors the running time depends only linearly on the number of buckets $k$ and the approximation factor is constant.

\subsection{Partitioning in Higher Dimensions}
\label{sec:partHighDim}

\subsubsection{Indexing To Find Maximum Variance}\label{sec:NewupdateVar}
We describe the core index $\maxVar$ that we use in all our partitioning algorithms for any dimension $d\geq 1$. All the details and technical proofs will be shown in Appendix~\ref{sec:updateVar}.
The exact description of the index depends on the type of aggregation queries we focus on. 

For SUM/COUNT queries,
we propose a simple index to find the query with the largest variance in a query rectangle. In particular, we build a dynamic range tree on $\samples$. Given a query rectangle $R$, we split it into two smaller rectangles $R_1, R_2$ such that $|R_1\cap \samples|=|R_2\cap \samples|=|R\cap \samples|/2$. Using a dynamic range tree~\cite{de1997computational} we return the rectangle $R_i$ (either $R_1$ or $R_2$) with the largest variance. We can show that $\nu_s(R_i)\geq \frac{1}{4}\optVar(R)$. The running time and the update time is $O(\log^d m)$. 

For AVG queries, the algorithm proposed in~\cite{liang2021combining} cannot be extended to the dynamic case. Hence we propose a new dynamic index with a better approximation factor. Similarly to ~\cite{liang2021combining}, we assume that every valid query that is contained in a bucket of the partitioning must contain at least $2\delta m$ samples (for a small parameter $\delta<1$), otherwise the estimation is not accurate.
For simplicity, we use the notation $\tilde{O}(\cdot)$ to hide $\log(m)$ factors.
In Appendix~\ref{sec:updateVar} we show the following crucial observation: for any rectangle $q$ inside a query rectangle $R$ with $|q\cap \samples|=\delta m$ that maximizes $\sum_{t\in q\cap \samples} t.a^2$, it holds that $\nu_s(q)\geq \frac{1}{4}\optVar(R)$. Hence, we build a dynamic index
so that given a query rectangle $R$ it returns a rectangle that contains $\delta m$ samples and the sum of squares of their aggregate values is close to the maximum sum.

We build a dynamic range tree $T'$ over the samples $\samples$, storing the number of samples in each node of the tree. 
Furthermore, we build another empty dynamic range tree $T$. We will use $T$ to store weighted rectangles (as points in $2d$) that contain at most $\delta m$ samples. More specifically, we store in $T$ the \emph{canonical} rectangles of $T'$ that contain at most $\delta m$ samples. Notice that there are $\tilde{O}(m)$ nodes in $T'$ hence $T$ uses $\tilde{O}(m)$ space.
When we have an insertion or deletion in $T'$ there are only $\tilde{O}(1)$ nodes/rectangles that are updated, hence we can update both $T'$ and $T$ in $\tilde{O}(1)$ time.
Given a query rectangle $R$ we use $T$ to find a rectangular query $q^*$ with the largest sum inside $R$ in $\tilde{O}(1)$ time. From the definition of a range tree, for any rectangle there is a partitioning of $\log^{d+1} m$ canonical rectangles from $T'$. Hence we can show that $\nu_s(q^*)\geq \frac{1}{4\log^{d+1} m}\optVar(R)$.
The exact complexities depend on the dynamic range tree structure we use; our data structure has roughly $O(m\log^{3d} m)$ space, $O(\log^{3d} m)$ update time, and $O(\log^{2d} m)$ query time.

\subsubsection{Partitioning}
We construct a partitioning by building a k-d tree using the dynamic procedure $\MAXV$ as we described above.
Here, we use our improved index $\maxVar$ to construct a k-d tree.
The high level description of the algorithm is:
\begin{enumerate}
    \item Max Heap $C$ containing partition $R_1$ covering all items in $\fullData$
    \item For $j=2$ to $k$
        \begin{enumerate}
            \item Extract partition $R_i$ with maximum $\maxVar(R_i)$ from $C$
            \item Create a partitioning of $R_i$ of two partitions $R_{i_1}$, $R_{i_2}$ by splitting on the median of $R_i$
            \item Insert $\maxVar(R_{i_1})$, $\maxVar(R_{i_2})$ in $C$
            \item Set $R_{i_1}$, $R_{i_2}$ as children of $R_i$ in $\PASStree$
        \end{enumerate}
\end{enumerate}

We can show that such a tree construction returns a partitioning which is near optimal with respect to the optimum partition tree construction following the same splitting criterion: split on the median of the leaf node with the largest maximum variance query. 
Overall we construct a data structure that can be updated in $O(\polylog m)$ time.
For a (re-)partition activation over a set $\samples$ of $m$ samples, we can construct a new $\PASStree$ with the following guarantees: For COUNT/SUM queries, $\PASStree$ can be constructed in $O(k\log^d m)$ time with approximation factor $2\sqrt{k}$.
For AVG queries, $\PASStree$ can be constructed in $O(k\log^{2d} m)$ time with approximation factor $2\log^{(d+1)/2} m$.

\subsection{Re-Partitioning Triggers}
\label{sec:re-partition}
Assume that the current partitioning is $\mathcal{R}$ and let $\maxVar(\mathcal{R})$ be the (approximate) maximum variance query with respect to the current set of samples $\samples$.
\sys first checks the number of samples in each bucket (leaf node) of the current $\PASStree$. If there is a leaf node $i$ associated with partition $R_i$ such that $|\samples_i|<<\frac{1}{\alpha}\log m$ (where $\alpha$ is the sampling rate) then there are not enough samples in $u$ to make robust estimators. Hence, we need to find a new re-partitioning. Even if the number of samples in each bucket is large our system might enable a re-partitioning:
For a partition $R_i$ in the leaf node layer of $\PASStree$ let $\maxVar_i=\maxVar(R_i)$ be the (approximate) maximum variance at the moment we constructed $\PASStree$. Let $\beta>1$ be a parameter that controls the maximum allowable change on the variance. It can either be decided by the user or we can set it to $\beta=10$. Assume that an update occurred in the leaf node associated with the partition $R_i$. After the update we get $\maxVar_i'=\MAXVar(R_i)$. 
If $\frac{1}{\beta}\maxVar_i\leq \MAXVar_i'\leq \beta\MAXVar_i$ then the new maximum variance in partition $b_i$ is not very different than before so we do not trigger a re-partition. Otherwise, the maximum variance in bucket $b_i$ changed by a factor larger than $\beta$ from the initial variance $\maxVar_i$.
In this case a re-partitioning might find a new tree with smaller maximum error.
We compute a new partitioning $\mathcal{R}'$ and hence a new tree $\PASStree$.
If $\maxVar(\mathcal{R}')<\frac{1}{\beta}\maxVar(\mathcal{R})$ then we activate a re-partition restarting the catch-up phase over the new tree $\PASStree$. On the other hand, if $\maxVar(\mathcal{R}')\geq \frac{1}{\beta}\maxVar(\mathcal{R})$ then our current partitioning $\mathcal{R}$ is good enough so we can still use it.
Of course, the user can also manually trigger re-partitioning. For example, the user can choose to re-partition once every hour, day, or after $\tau$ insertions and deletions have occurred. In Appendix~\ref{appx:re-partition}, we also describe how \sys can execute either partial or full re-partitioning.

\reviewtwo{
\subsection{Discussion: Selection of Parameters and Single Synopsis}
\label{subsection:discussion}
}
\newchanges{
Given a memory constraint, our system defines the parameters $m, k$, i.e. the number of samples and the number of leaf nodes in \syn. The asymptotic space of our index, as we saw in the previous subsections, is roughly $O(m+k)$, skipping $\log m$ factors. However, $k\ll m$ so the asymptotic space of our system is $O(m)$. In particular, in our experiments we observed that choosing $k\approx \frac{0.5}{100}m$ always gives a low space and efficient data structure with low error guarantees.
Hence, using the memory constraint, the space of our index with respect to $m, k$, and the relation between $m, k$ our system can automatically set these parameters.
If the user also gives a query processing constraint, we might change the values of $m, k$ to satisfy this requirement. In particular, the query time of our index is $O(k+\min\{\log^d m, m^*\})$, where $m^*$ is the maximum number of samples in a leaf node, which is typically no more than $2\cdot m/k$.  
}

\reviewtwo{
So far, we consider that the user defines different synopses for different query templates.
We recall that for a query template the user should give three inputs: i) the type of the query SUM/COUNT/AVG/MAX/MIN, ii)the aggregation attribute $A$, and iii) the predicate attributes $c_1,\ldots, c_d$. Any combination of these inputs creates a different template.
We propose two ways to handle multiple templates. The first one has the same theoretical guarantees with respect to the maximum confidence interval as we had for a single synopsis.
In the second one, we present a simpler heuristic way to handle multiple query templates. We can still have confidence intervals but we do not have any guarantee over them.
}

\reviewtwo{
First method: Recall that a synopsis consists of the partition tree \syn and a set of at most $2m$ samples $\samples$. As we described previously the overall space for a synopsis is roughly $O(m+k)=O(m)$.
We describe a simple method having one global set $\samples$ of at most $2m$ samples and for each query template we only need to construct a different partition tree \syn. Hence, if the user is interested in $L$ different query templates the total space of our index is $O(m+L\cdot k)$.
The main idea is that we store $\samples$ only once in a dynamic range-tree or a k-d tree and construct one partition tree $\syn_h$ with $k$ leaf nodes for each query template $h\leq L$. Each node of every $\syn_h$ stores the (updated) statistics and the associated (hyper-)rectangle of the corresponding $h$-th partitioning. The only difference now is that we do not store the samples in the leaf nodes of each tree. Instead, whenever we need access to the samples in a leaf node $u$, we run a reporting query with the corresponding hyper-rectangle $R_u$ in the range tree (or k-d tree). In this way, using only $O(m+L\cdot k)$ space, we have theoretical error guarantees for every different query that belongs in one of the supported query template. The update time increases by a factor of $L$, however, all update operations in a tree are extremely fast and they can be executed in parallel for different trees.
Furthermore, we note that ideally in our system we would like to know all query templates upfront so that the corresponding tree is ``ready'' when we get a query from a supported template. However, even if the templates are not known upfront, when we see a query from a new template we can construct a new partition tree using our efficient partitioning algorithms in roughly $O(k\cdot\polylog m)$ time. Then we start the catch-up phase only for this tree and start answering queries from this new template with error guarantees.
Finally, we notice that the user can give an overall memory constraint, and not one memory constraint for each different synopsis. Using the overall upper bound, our system can automatically decide and update the number of samples we store or the size of the partition trees that we maintain in order to satisfy the memory constraint.
}

\reviewtwo{
Second method: There is also a simpler, heuristic way to handle queries from different templates using only one partition tree \syn. Imagine that we construct a partition tree using the SUM aggregation function over the aggregation attribute $A$ and predicate attributes $c_1, \ldots, c_d$. If the user asks a query with a different type of aggregation function, say AVG or COUNT, we can still use the same tree to answer the query as long as we store and maintain the sum and the count in each node of the tree. If the user asks a query with a different aggregation attribute $B$ instead of $A$ there are two ways to handle it: i) We store and maintain statistics in \syn for all attributes including $B$ (we always assume constant number of attributes). Hence, we can still answer a query using the partition tree straightforwardly. ii) If we do not store and maintain statistics for all attributes, then we can perform stratified sampling using the strata and the samples in the leaf nodes of the \syn.
Finally, assume that the system encounters a new query having different predicate attributes. There are a few ways to handle it: i) Assume that the dynamic partition tree \syn we constructed at the beginning was constructed over all possible predicate attributes of the tuples. Recall that throughout our submission we always assume a constant number of attributes. We can always query such tree with any query range over any subset of predicate attributes and return an estimation with confidence intervals. The tree is not optimized strictly with respect to the predicate attributes of the query (as in the first method) however we expect the error to be low. ii) Simply apply uniform sampling using the samples $\samples$, or iii) re-partition the tree for the new predicate attributes.
}

\section{Experiments}
\label{sec:experiments}
We run our experiments on a Linux machine  with an Intel Core i7-8700 3.2GHz CPU and 16GB RAM. 

\subsection{Setup}
We generate query workloads of 2000 queries by uniformly sampling from rectangular range queries over the predicates. 
We then initialize a \sys instance with a user-specified sample rate, a catch-up ratio and a number of leaf nodes of the partition tree to compare with other baselines (these parameters directly control the Throughput, Query Latency, and Storage Size).

\subsubsection{Datasets}\sloppy
\textsf{Intel Wireless dataset}. The Intel Wireless dataset \cite{intelwireless_} contains 3 million rows of sensor data collected in the Berkeley Research lab in 2004. Each row contains measurements like humidity, temperature, light, voltage as well as the date and time each record was collected. 

\noindent \textsf{New York Taxi Records dataset}. The New York City Taxi Trip Records dataset \cite{nyctaxi} contains 7.7 million rows of yellow and green taxi trip records collected in January 2019. Each record contains information about the trip including
pickUpDateTime, dropOffDateTime, tripDistance, dropOffLocation, passengerCount, etc.

\noindent \textsf{NASDAQ ETF Prices dataset}. The NASDAQ Exchange Traded Fund (ETF) Prices dataset\cite{nasdaq} contains 2166 ETFs traded in the NASDAQ exchange from April 1986 to April 2020. There are 4 million entries in the dataset and each entry contains the date, the volume of transactions of an ETF on the date, and 4 prices: the price of an ETF when the market opens and closes; the highest and the lowest of its daily price range.

\subsubsection{Metrics and ground truth}
We report the wall-clock latency and the throughput, i.e. number of requests (query/data) processed per second. 
To measure the accuracy of the system,  we report the 95 percentile of the relative error which is the difference between ground truth and estimated query result divided by the ground truth. 
We define the ground truth to be w.r.t all the tuples available when the query arrives, i.e. the true results reflect all insertions and deletions up to its arrival point.
To make sure our experiments are deterministic, we fix this sequence up-front and ensure they are the same for each baseline.

\subsubsection{Baselines} 
All of these baselines are tuned to roughly control for query latency.

\vspace{0.5em}

\noindent \textbf{Reservoir Sampling (RS) and Stratified Reservoir Sampling (SRS).} We construct a uniform sample of the entire data set which is maintained using the reservoir sampling algorithm~\cite{vitter1985random}. We use a variant of RS first designed for the AQUA system that handles both insertions and deletions~\cite{gibbons2002fast}\footnote{Due to its age, a direct comparison with AQUA was not feasible}. Unless otherwise noted, we use a 1\% sample of data. For stratified seservoir sampling, the strata is constructed using a equal-depth partitioning algorithm.

\noindent \textbf{DeepDB. }  We also compare with a machine learning-based baseline called \textsf{DeepDB}\cite{hilprecht2019deepdb}. DeepDB achieves state-of-the-art AQP results in the static setting, and we chose it as a baseline since it has limited support for dynamic data. In our baseline, DeepDB trains on 10\% of the data. We set this to be equivalent to the ``catch-up'' sampling in \syn.

\noindent \textbf{Dynamic Partition Tree-Only (\syn). } We compare with a baseline of only using a single \syn synopsis without online optimization. This synopsis is constructed once and then used for the duration of the experiment. Unless otherwise noted there are 128 leaf nodes in a balanced binary tree, the leaf nodes are associated with 1\% samples of their respective strata, and the catch-up sampling rate 10\% of the data.

\noindent \textbf{\sys. } Finally, we evaluate the full-featured \sys system. This includes a \syn and also performs re-partitioning if needed. Unless otherwise noted there are 128 leaf nodes in a balanced binary tree, the leaf nodes are associated with 1\% samples of their respective strata, and the catch-up sampling rate 10\% of the data.

\vspace{0.5em}

The storage costs of the baselines on the NYC Taxi dataset given the typical setting (128 leaf-nodes, 10\% catch-up rate, and 1\% sample rate) are: reservoir sampling baseline takes about 5MB, \sys and \syn takes about 6MB, a DeepDB baseline trained with 10\% of the data is about 60MB.

\subsection{Accuracy}
\label{sec:end-to-end}


\begin{table*}[ht]
    \caption{Median relative error (\%) of 2000 SUM random queries and average query latency (ms/query) over three datasets.}
    \label{tbl:macro-acc-latency-vs-deep-rs}
    \centering
    
    \begin{tabular}{|c|}
        \hline
        \multicolumn{1}{|c|}{---}\\
        \hline
        Approach \\\hline
        \sys   \\ 
        DeepDB   \\ 
        RS    \\
        SRS    \\\hline
    \end{tabular}
    \begin{tabular}{|c|c|c|}
        \hline
        \multicolumn{3}{|c|}{Intel (\%)}\\
        \hline
        0.2 & 0.5 & 0.9\\\hline
        0.67  & 0.62 & 0.33 \\        
        1.5 & 1.7 & 0.8\\
        2.1 & 1.6 & 1.3  \\
        1.3 & 1.3 & 1.2  \\\hline

    \end{tabular}
    \begin{tabular}{|c|c|c|}
        \hline
        \multicolumn{3}{|c|}{NYC (\%)}\\
        \hline
        0.2 & 0.5 & 0.9\\\hline
        0.48  & 0.22 & 0.2 \\        
        4.7 & 4.7 & 4.7\\
        3.4 & 2.1 & 0.94  \\
        2.4 & 1.2 & 0.95  \\\hline

    \end{tabular}
    \begin{tabular}{|c|c|c|}
        \hline
        \multicolumn{3}{|c|}{ETF (\%)}\\
        \hline
        0.2 & 0.5 & 0.9\\\hline
        5  & 4.3 & 2.3 \\   
        -  & - & - \\   
        16  & 9.8 & 8.6 \\   
        10  & 8.2 & 8 \\\hline

    \end{tabular}
    \begin{tabular}{|c|c|c|}
        \hline
        \multicolumn{3}{|c|}{Intel (ms/query)}\\
        \hline
        0.2 & 0.5 & 0.9\\\hline
        0.19  & 0.31 & 0.63 \\        
        0.6 & 0.6 & 0.6\\
        2.5 & 6.3 & 13.2\\
        3.1 & 6 & 10.7  \\\hline

    \end{tabular}
    \begin{tabular}{|c|c|c|}
        \hline
        \multicolumn{3}{|c|}{NYC (ms/query)}\\
        \hline
        0.2 & 0.5 & 0.9\\\hline
        0.27  & 0.57 & 0.97 \\        
        0.6 & 0.6 & 0.6\\ 
        4.7 & 14.2 & 30.6\\
        4.6 & 14.7 & 25.3  \\\hline

    \end{tabular}
    \begin{tabular}{|c|c|c|}
        \hline
        \multicolumn{3}{|c|}{ETF (ms/query)}\\
        \hline
        0.2 & 0.5 & 0.9\\\hline
        0.14  & 0.28 & 0.46 \\        
        0.6  & 0.6 & 0.6 \\   
        2.58  & 6.8 & 13 \\   
        2.66 & 5.2 & 12.7  \\\hline

    \end{tabular}
    

\end{table*}

We first evaluate the end-to-end performance of \sys and the baselines on a 1d problem (1 predicate attribute).
For the NYC Taxi dataset, we use the \texttt{pickUpTime} attribute as the predicate attribute and the \texttt{tripDistance} attribute as the aggregate attribute; for the ETF dataset, we use the \texttt{volume} attribute as the predicate attribute and the \texttt{close} attribute as the aggregate attribute; for the Intel Wireless dataset, we use the \texttt{time} and \texttt{light} attributes as predicate and aggregate attribute respectively.

We start with 10\% of the data in Kafka which is used by the baselines for initialization (simulating historical data). We incrementally add 10\% more data in increments (simulating new data arrival).
After every 10\% increment, we re-train the model for DeepDB and re-initialize the DPT used by \sys. 
We report results when 20\%, 50\%, and 90\% of the rows from each dataset are inserted into the system. The median relative error and the corresponding average query latency can be found in Table \ref{tbl:macro-acc-latency-vs-deep-rs}.

We can see that \sys has the overall best accuracy while controlling for query latency.
We note that the accuracy of DeepDB is stable as a function of progress. This is because as a learned model DeepDB has a roughly fixed resolution of the data (it does not increase the number of parameters as more data is inserted).\footnote{We omit the results of DeepDB on the ETF dataset in Table \ref{tbl:macro-acc-latency-vs-deep-rs} due to very large error ($>$ 1000\%) for SUM queries while the error of COUNT queries is reasonable.} These findings are consistent with results from \cite{liang2021combining}. The accuracy of RS and SRS improves at a cost of a higher query latency. 

\subsection{Online Performance}
 We populate Kafka with the first $p$ percent of the NYC Taxi dataset ($p$ varies from 10 to 90). Like before, we initialize \sys on the first 10\% of data and then incrementally add increments of 10\% more. In this experiment, we construct a mixed update workload of both insertions and deletions.
On the left plot of Figure \ref{fig:thruputopt}, we show the throughput of handling insertions and deletions using a pool of 12 threads. We can see the performance of \sys is quite stable and does not change with the size of existing data or the amount of data that have been processed. For each insertion and deletion, we simply find the target node in $O(\log(k))$ and modify the summary. Even though a larger reservoir size increases the overhead of manipulating the samples for reservoir sampling, the increased overhead is unnoticeable. This is because the stratum stored in each node is $\frac{1}{k}$ of the reservoir, and each stratum is independent with others and race condition only happens if two workers are working on the same node. 

On the right plot of Figure \ref{fig:thruputopt}, we show the re-optimization time cost in seconds by \sys and DeepDB. The cost to initialize \sys increases with the number of tuples stored in Kafka but it is still much cheaper than DeepDB. It is worth noting that the re-optimization cost of DeepDB is the cost of re-training instead of incremental training. This is mostly due to the constraint of the API exposed by DeepDB, and we observe that re-train a model with 2$n$ samples is faster than train a model with $n$ samples then incrementally train another $n$ samples.
The results suggest that complex, learned synopses are not ideal in the dynamic setting.

\begin{figure}[t]
  \centering
  \includegraphics[width=0.8\linewidth]{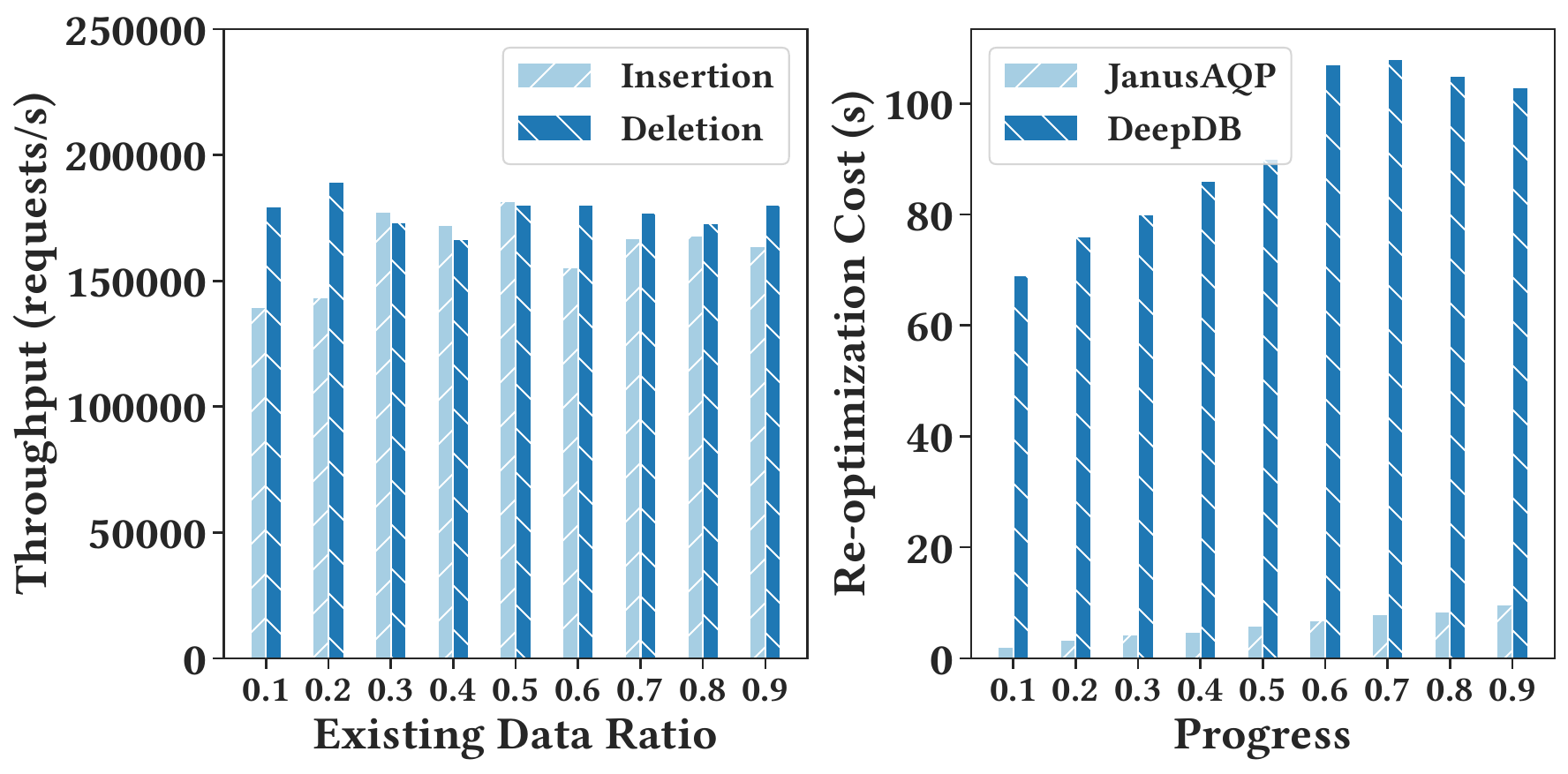}
  \caption{We evaluate the throughput of \sys when handling insertions and deletions in multi-threaded mode. We also compare the re-optimization cost with DeepDB.}
  \label{fig:thruputopt}
\end{figure}

\subsection{Handling Deletion}
We construct a \sys instance with the 50\% percent of each dataset, then we delete the last $p$\% of data of the first 50\% ($p$ varies from 1\% to 9\%). After \sys process all the deletions, a query workload of 2000 random queries is evaluated and we record the median relative error of the 2000 queries. We use the data that remains in the system to compute the ground truth, e.g., for $p=1\%$, the ground truth is computed with the first 49\% of each dataset. 

Results can be found in Figure \ref{fig:del3in1}, we notice that the relative error is relatively stable when we vary the deletion percentage. This is because the tuples that are being deleted are uniformly distributed over the predicate attributes of the query workloads, i.e. the deletion would occur in each leaf node of the DPT with roughly the same probability, therefore, the DPT without re-optimization works reasonably well.
In another experiment we artificially generate deletions that are skewed to demonstrate scenario where re-optimization is needed, details can be found in Sec. \ref{sec:exp-del}.

\begin{figure}[t]
  \centering
  \includegraphics[scale=0.22]{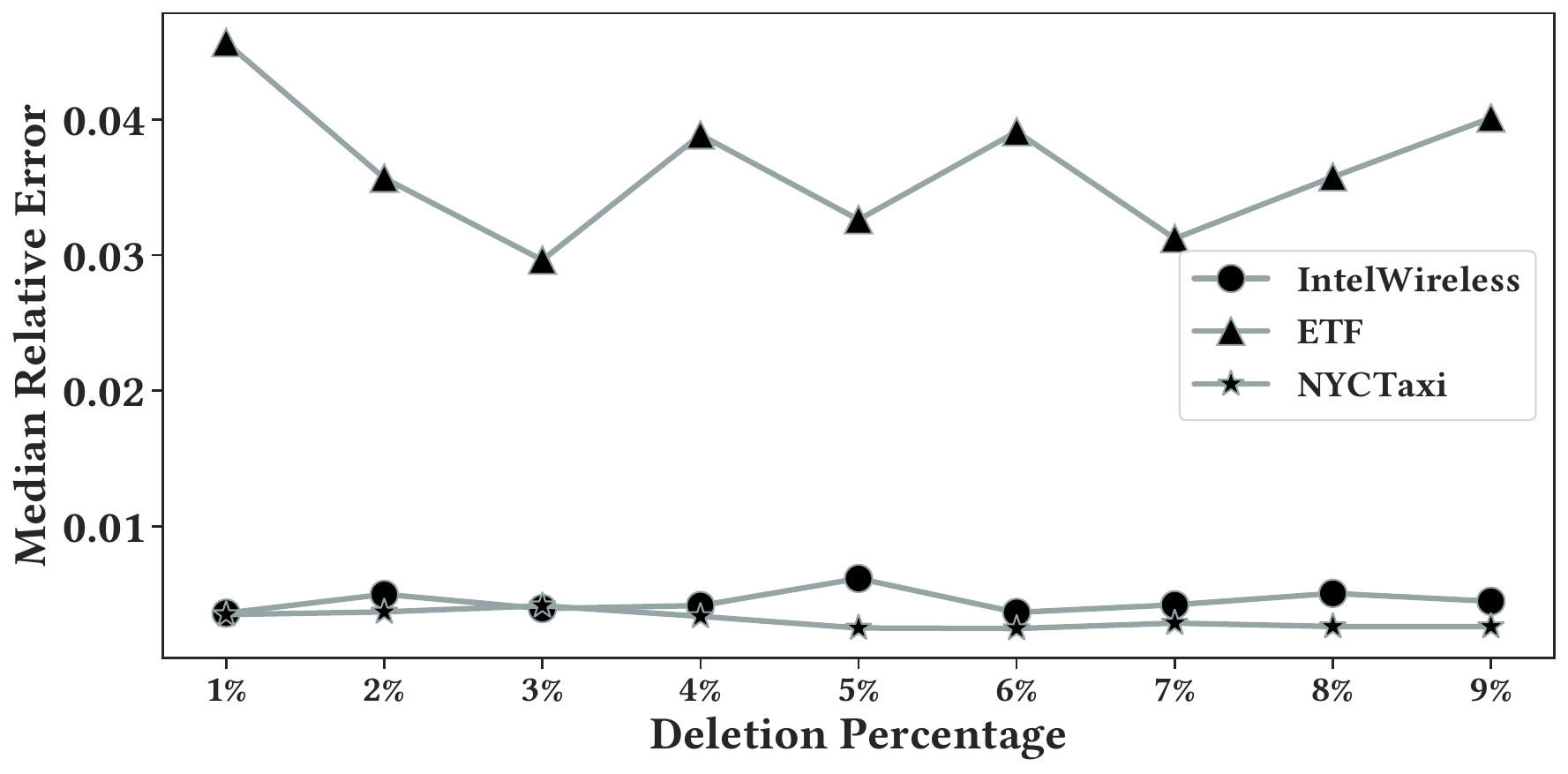}
  \caption{Median relative error of \sys varying the amount of deletions from 1\% to 9\% over three datasets.}
  \label{fig:del3in1}
\end{figure}

\subsection{The Catch-up Phase}
In this experiment, we study how the catch-up phase can impact the accuracy and performance of the entire system.

\subsubsection{Accuracy}
We use the entire Intel wireless dataset as the existing data. We compare a set of \sys(128, $c$, 1\%) instances where the catch-up goal $c$ varies from 1\% to 10\% with a step of 1\%. When each \sys instance reaches the catch-up goal, we use it to evaluate the same set of 2000 random queries generated using the light attribute as the aggregate attribute and the time attribute as the predicate attribute.

The results can be found in the left plot of Figure \ref{fig:intel-catchup-ratio-vs-nx}. As a reference, we also show the accuracy of an RS baseline with 1\% sample rate. We notice that \sys(128,1\%,1\%) has no advantage against the RS baseline because neither the samples nor the summaries built during catch-up could provide better accuracy. As we increase the catch-up ratio, we can see an improvement in accuracy because the quality of the summaries built by the catch-up phase improved. Comparing with the expensive offline pre-processing used in \cite{liang2021combining}, we believe the catch-up phase is a better alternative that provides another knob to tune the tradeoff between accuracy and cost.

\begin{figure}[t]
  \centering
  \includegraphics[width=0.8\linewidth]{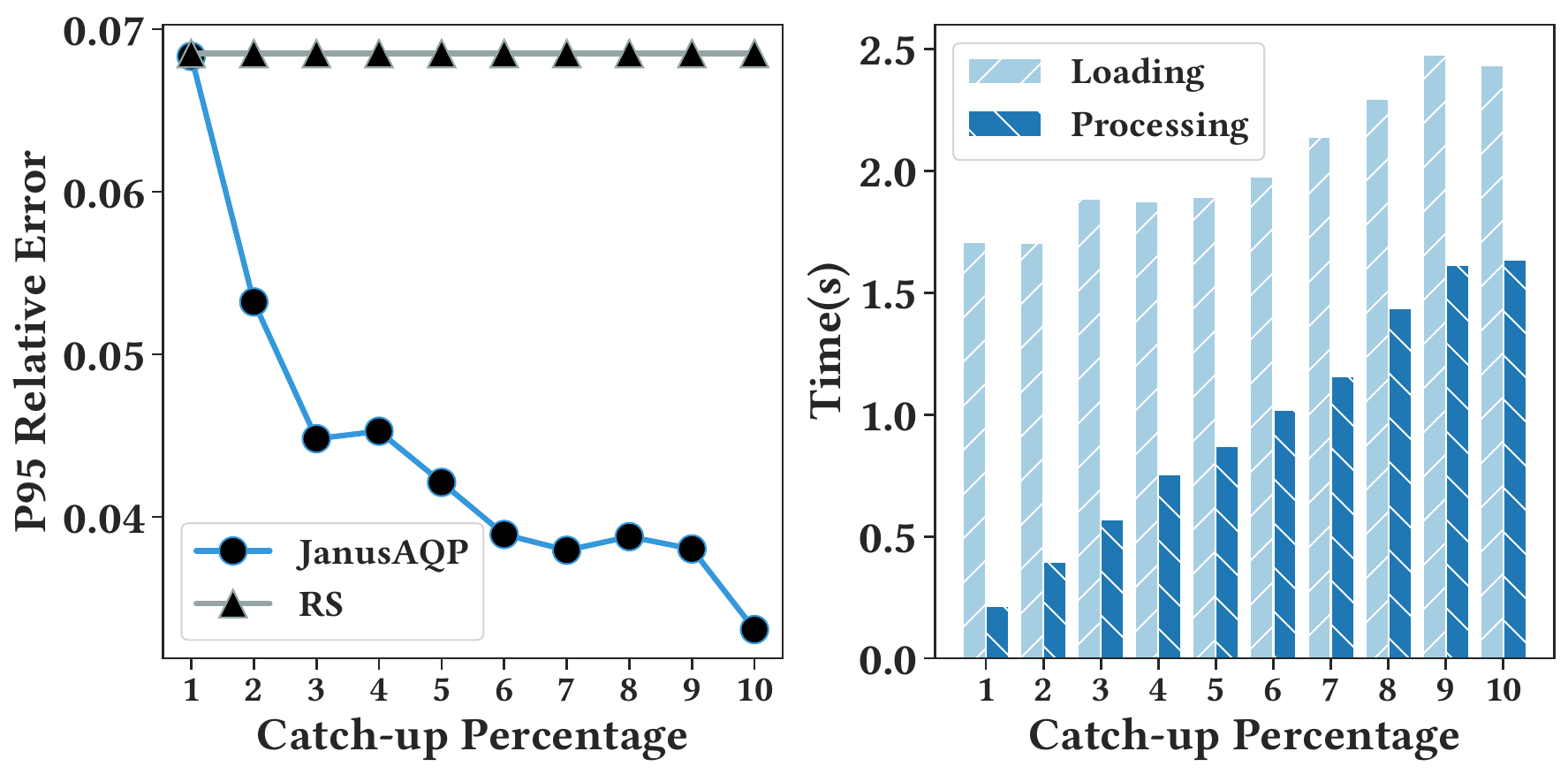}
  \caption{Varying the catch-up goal from 1\% to 10\% of the data, we evaluate the accuracy of \sys (left plot) and the time cost of the catch-up phase (right plot).}
  \label{fig:intel-catchup-ratio-vs-nx}\vspace{-1em}
\end{figure}

\subsubsection{Overhead}
The overhead of the catch-up phase comes from two sources: the loading and processing of the samples. 
We distinguish and measure the two types of overhead in terms of their time cost. Data loading time measures the time spent on calling the Kafka \code{poll()} API, transferring the data, and ETL operations that are necessary to prepare the data for \sys to process. It is worth noting that the data loading cost is part of the essential cost that occurs in all systems and is usually less relevant to the core design of the system but more relevant to the design of interfaces. For example, with a different interface, instead of dealing with the strings from Kafka that can be expensive to parse, the system could use Protocol Buffers\cite{protobuff} for more efficient data exchanging or even offload some of the ETL duties to the client-side as described in \cite{ding2021ciao}. On the other hand, the data processing time stands for the time taken by \sys to analyze the data then accordingly modify internal data structures that will be used for query processing.  

Results can be found in the right plot of Figure \ref{fig:intel-catchup-ratio-vs-nx}. We can see that the data processing with a single thread takes less than 1.5 seconds for a catch-up ratio of 10\%, which is equivalent to a throughput of processing 160,000 tuples per second. Furthermore, the data loading cost is much higher than the data processing cost and we believe the data loading cost can be further improved by more engineering efforts and techniques such as client-assisted data loading\cite{ding2021ciao}.

\reviewtwo{
\subsection{Dynamic Query Templates}
\label{subsection:robustness}
}

\reviewtwo{
In this section, using the NYC Taxi dataset, we conduct experiments to demonstrate the robustness of \sys. In particular, we implement the heuristic way of \sys to handle new query templates, as described in Subsection~\ref{subsection:discussion}. We explore the error considering all three cases: Different predicate attributes, different aggregation attribute, and different aggregation function.
}

\reviewtwo{
First, we consider the case where \sys is constructed using one predicate attribute but the query uses another attribute as predicate attribute. In Subsection~\ref{subsection:discussion} we proposed three ways to handle it. Here we implement the simpler way where we apply uniform sampling using the samples $\samples$ which leads to higher query latency similar to what we observed for RS in Table \ref{tbl:macro-acc-latency-vs-deep-rs}.  The results can be found in the left plot of Figure \ref{fig:p95-3in1}: the PickupOverPickup baseline is the case where the query and \sys both use the PickupTime attribute as predicate attribute; the DropoffOverDropoff baseline is the case where the query and \sys both use the DropoffTime attribute as predicate attribute; lastly, the DropoffOverPickup baseline is the case where the query use the DropoffTime attribute as predicate attribute while \sys is constructed using PickupTime as predicate attribute. We observe that when  the two attributes differs (DropoffOVerPickup) we have the highest error (but still it happens to be quite competitive). After a re-partitioning with respect to the new predicate attribute the high accuracy can be restored (DropoffOverDropoff). Even though we drew three curves over the entire progress, we assume that the system first considers PickupOverPickup between $[0.1-0.3)$ (cycles), then DropoffOverPickup between $[0.3-0.6)$ (stars), and finally DropoffOverDropoff between $[0.7-0.9]$ (triangles). We present all values over the three curves to show all the details. This is also the case in all plots of Figure~\ref{subsection:discussion}.
}


\reviewtwo{
Second, we consider the case where \sys is constructed using one aggregation attribute but the query uses a different aggregation attribute. In Section~\ref{sec:partAlg}, we described how the aggregation attribute is used by our partitioning algorithm while constructing \sys. In this experiment, we show how the change of aggregation attribute affects the performance of \sys. The result can be found in the middle plot of Figure \ref{fig:p95-3in1} where we show two baselines indicating whether the aggregation attribute used by \sys is the same with the aggregate attribute of the query. We can see the accuracy of the two baselines are quite close, which suggests that \sys is quite robust over the change of the aggregation attribute.
}


\reviewtwo{
Finally, we study how \sys performs over different aggregation functions. As we described earlier, SUM, CNT and AVG are all part of the statistics we keep and maintain in the partition tree therefore no additional effort is required to take care of a change in aggregation function. Results in the right plot of \ref{fig:p95-3in1} suggest all the three aggregate functions can achieve good accuracy, which also implies other aggregate functions such as STDDEV that can be composed using SUM and CNT would also perform well.
}

\reviewtwo{Overall, we show that our system can support different query templates with low error even with the heuristic way that stores and maintains only one partition tree.}

\begin{figure}[t]
  \centering
  \includegraphics[width=1.0\linewidth]{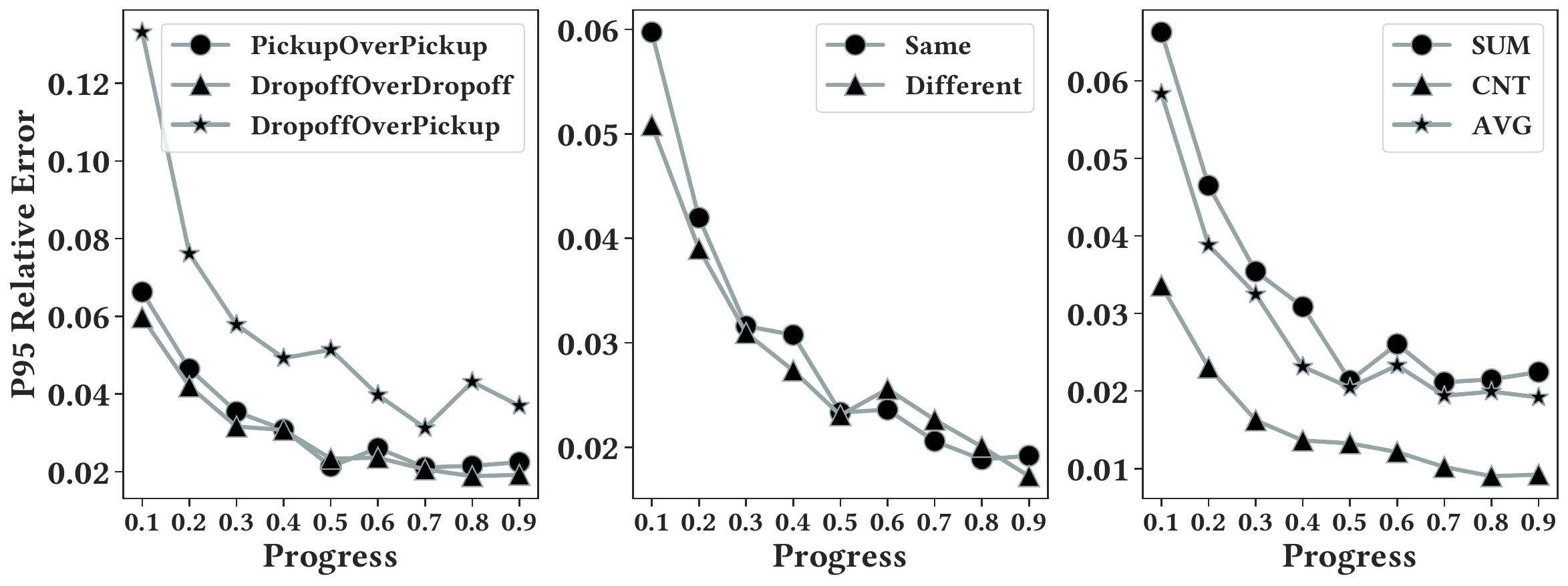}
  \caption{\reviewtwo{We evaluate the robustness of \sys over three scenarios where the predicate attribute, the aggregate attribute, and the aggregate function used by the query workload might change from the ones that are used to construct \sys.}}
  \label{fig:p95-3in1}\vspace{-1em}
\end{figure}

\subsection{Multi-dimensional Query Templates}
Next, we investigate the performance of \sys with multi-dimensional queries on the NASDAQ ETF Prices dataset. We randomly generate 2000 queries from a 5-D query template that uses the volume attribute as the target attribute, the date attribute and the 4 price attributes as predicate attributes. We perform the same workflow as we did in Section \ref{sec:end-to-end}. We first compare the median relative error of \sys(256,10\%,1\%) with DeepDB and the results can be found in the left plot of Figure \ref{fig:etf-kdtree}. We notice that the accuracy of \sys is better than DeepDB but the relative error increases for both. This is because multi-dimensional queries are usually more selective. Also, because the queries are generated using the entire dataset, we notice that many of the ground truths generated using the first 20\% of the data are 0s. Therefore, in the experiment, we start with 30\% of the data. 
On the right plot of Figure \ref{fig:etf-kdtree}, we can find the re-optimization cost of \sys is lower than DeepDB but is more expensive than in the 1D setting. While the increase of dimensions can indeed make it more expensive to process the samples we fetched during catch-up, we believe the re-optimization cost can be further improved with more engineering efforts. 

\begin{figure}[t]
  \centering
  \includegraphics[width=0.8\linewidth]{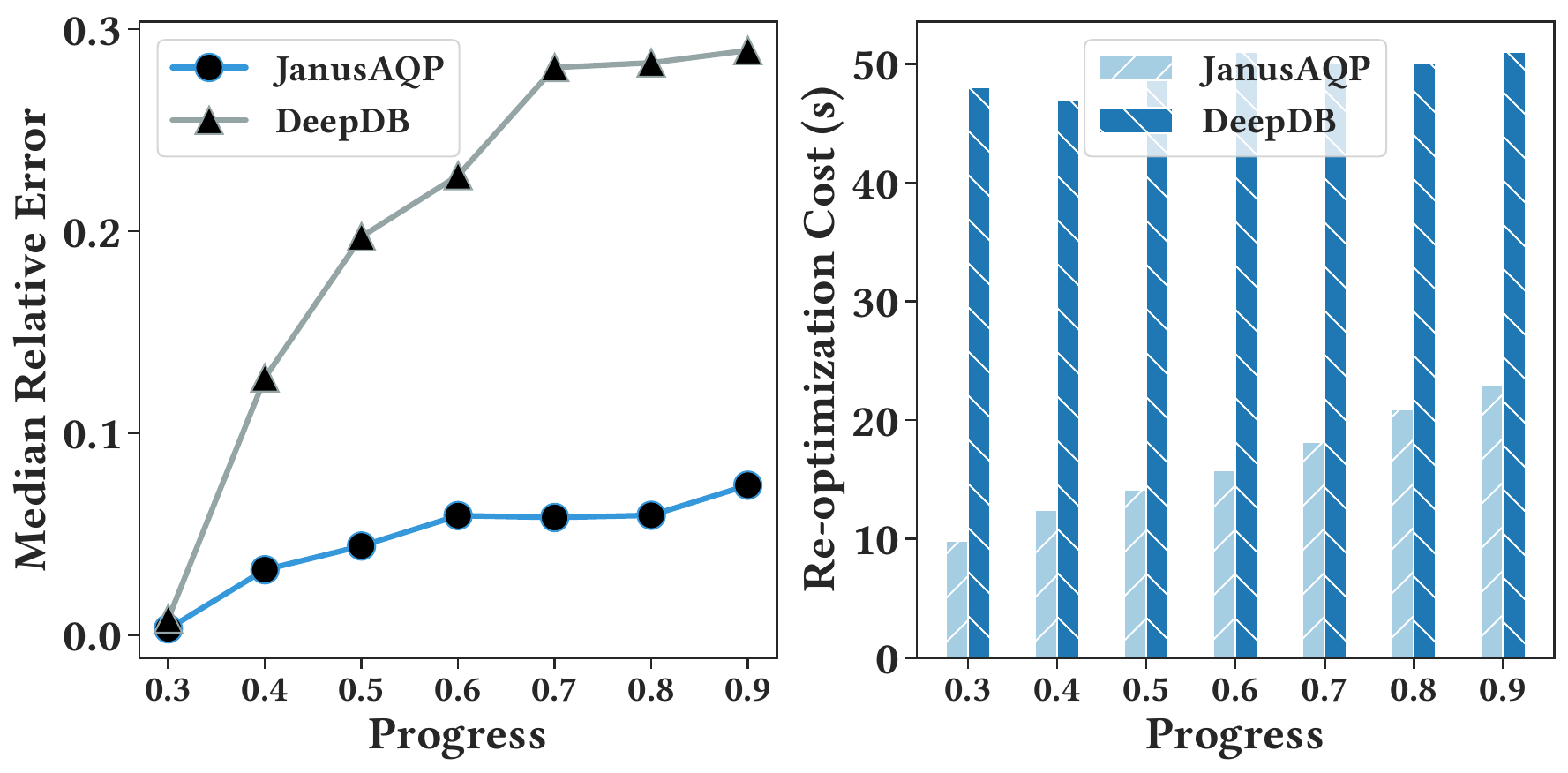}
  \caption{We compare the median relative error and the re-optimization cost of \sys with DeepDB on multi-dimensional queries. 
  }
  \label{fig:etf-kdtree}
\end{figure}

\subsection{Re-partitioning}
\label{sec:exp-del}
In the first experiment, using the NYC Taxi dataset, \sys performs a periodic re-partitioning after every 10\% insertions. For comparison, the \syn baseline does not perform any re-partitioning and we evaluate the accuracy. 
We deliberately skew the insertions by sorting on \texttt{pickUpDateTime} so that new insertions would hit a small number of partitions.
The results are illustrated in Figure \ref{fig:repartition}(left), we can see the relative error of \syn increases drastically due to a partition tree that becomes more and more imbalanced with new insertions. With periodic re-partition, \sys keeps the accuracy at a controlled level.

\begin{figure}[t]
  \centering
  \includegraphics[width=0.8\linewidth]{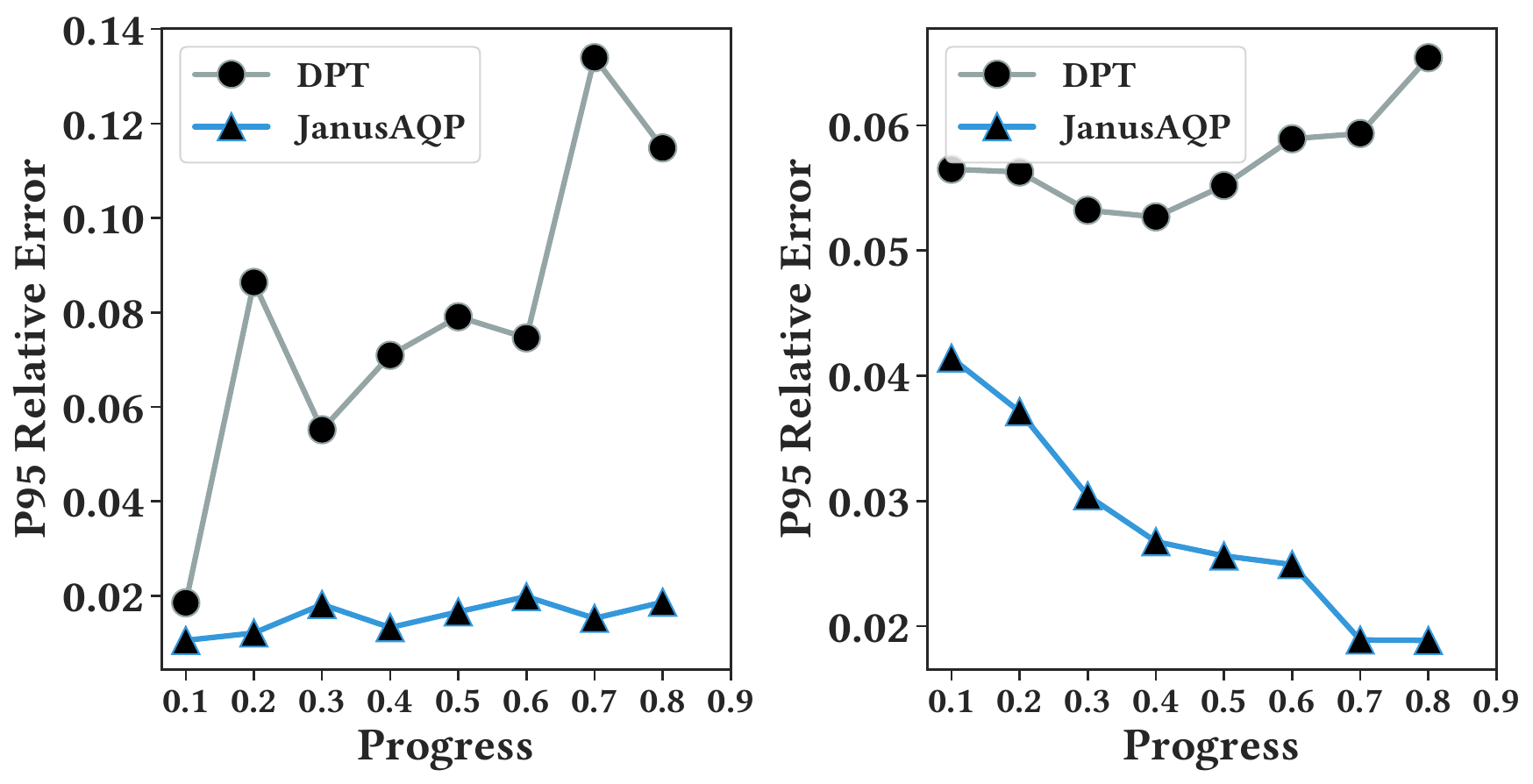}
  \caption{We compare the accuracy of \sys and \syn in two scenarios that cause imbalanced partition trees.}
  \label{fig:repartition}
\end{figure}

In the second experiment, we use the \texttt{pickupTimeOfDay} as the predicate attribute. Because the dataset is randomly distributed
over the \texttt{pickupTimeOfDay} attribute, the insertions are not skewed as in the previous setting. 
To demonstrate a situation where a re-partition is triggered by deletions, we randomly choose 10\% of the nodes and we randomly delete half of the samples that belong to these nodes then we insert the next 10\% data. After the insertion, the re-partition will be triggered for \sys. For comparison, we use a \syn baseline that does not perform any re-partition. We perform the same operations to the leaf nodes of the \syn baseline then we evaluate the same set of queries.
The results can be found in the right plot of Figure \ref{fig:repartition}, we can see the relative error of \syn increases due to the imbalanced partition tree while the error of \sys drops because of re-partition.



\subsection{A More Efficient Partitioning Algorithm}

In Section \ref{sec:partAlg}, we propose a binary search-based (BS-based) partitioning algorithm for 1 dimension that is much more efficient. In this experiment, we compare the accuracy and time cost of the BS-based algorithm with the dynamic programming-based partitioning algorithm used by PASS on the Intel Wireless dataset. We implement the BS-based algorithm in Python in our code base of PASS for a fair comparison. We measure the time cost in seconds of each partitioning algorithm given different number of partitions, we also compare the median relative error of the PASS variation over 2000 randomly generate queries.

\begin{table}[t]
\caption{We compare our new binary search-based (BS) partitioning algorithm with the dynamic programming-based (DP) algorithm proposed in~\cite{liang2021combining} on the Intel dataset.}
\label{tbl:partition}
\footnotesize
\centering
\begin{tabular}{|cl|l|l|l|l|}
\hline
\multicolumn{1}{|l}{} &      & 16    & 32    & 64     & 128    \\ \hline
\multirow{2}{*}{Partition Time (s)}       & DP & 16    & 22    & 382    & 6349   \\
                                          & BS  & 0.3   & 0.3   & 0.4    & 1.6    \\
\multirow{2}{*}{Median RE (CNT)}        & DP & 0.2\% & 0.1\% & 0.05\% & 0.04\% \\
                                          & BS  & 0.6\% & 0.4\% & 0.1\%  & 0.1\%  \\
\multirow{2}{*}{Median RE (SUM)}          & DP & 0.2\% & 0.1\% & 0.07\% & 0.05\% \\
                                          & BS  & 1\%   & 0.9\% & 0.2\%  & 0.2\%  \\
\multirow{2}{*}{Median RE (AVG)}          & DP & 0.2\% & 0.1\% & 0.08\% & 0.05\% \\
                                          & BS  & 1\%   & 0.7\% & 0.2\%  & 0.15\% \\ \hline
\end{tabular}
\end{table}
The result can be found in Table \ref{tbl:partition}. We vary the number of partitions from 16 to 128, as we increase the number of partitions, the sample size used by the algorithms also increase. We notice that the time cost of the DP-based algorithm increase drastically with the number of partitions while the time cost of the BS-based algorithm increase slightly.
On the accuracy side, the DP-based algorithm does lead to a lower error but the BS-based algorithm also introduce good accuracy. Overall, we believe the BS-based algorithm is more scalable than the DP-based algorithm and it provides favorable trade-off between cost and accuracy.

\bibliographystyle{abbrv}

\bibliography{ref}

\newpage
\appendix
\section{Sampling from Kafka-like Systems}
\label{sec:kafka-sampler}
Random sampling is a key building block of \sys because it is used in many components: we use random samples to build the leaf layer of the partition tree and initialize the reservoir that are later used to solve queries. During the catch-up phase, we also use random samples to construct the summaries that are stored in the leaf nodes. Random sampling affects both the accuracy and performance of the entire \sys system: a biased sample hurts the accuracy of the system and expensive sampling operations leads to a higher latency and lower throughput.

Designing a random sampler for message brokers like Kafka can be a non-trivial task because the API of such systems usually does not provide random access to the data. To retrieve data from a Kafka topic, a Kafka consumer has to send an offset to the server indicating the location it wants to access data from. Therefore, a naive random sampler can be implemented by using the \code{poll()} API with a random offset. 

However, such a naive implementation can be expensive because each poll could retrieve a batch of thousands of tuples that are contiguous (therefore biased). To guarantee unbiased sampling, we will have to keep only a small portion from each batch and discard most of the tuples, set another random offset, repeat the process until we collected enough samples. 

To build a scalable, efficient and unbiased random sampler for Kafka, we need more control over the polling process besides the offset, to be more specific, we want to control the size of each poll. We propose two sampling methods:

\textbf{Sequential Sampler} A sequential sampler retrieves the data in a sequential manner. In each poll, a random sample is drawn from the batch and the rest tuples are discard. Sequential samplers work the best when the size of the dataset is medium to large because the entire dataset is transferred from Kafka to the sampler. Therefore, one overhead of this approach is the network traffic of transferring the entire dataset, another drawback of a sequential sampler is that the random samples are only available until the sampler have retrieved the entire dataset, the high latency might not be acceptable in some situations.

\textbf{Singleton Sampler} In each poll, a singleton sampler request one tuple from a random offset, it repeats until enough sample has been collected. Singleton samplers minimize the network traffic with the cost of server-side overhead due to a more frequent use of the API. The main advantage of a singleton sampler is that it offers lower latency because the random sample is built incrementally, a small random sample can be available with a lower latency. 

In general, we observe that, for light-weight sampling tasks or tasks that requires low latency, a singleton sampler works the best. In scenarios where the dataset size is medium to large and the expected latency is acceptable, the sequential sampler might be preferred for a more consistent performance. There might be an interesting mid-ground for sampling from Kafka-like systems that can offer better efficiency and latency and we plan to investigate in a future work.

\mparagraph{Experiments}
\label{sec:sampler-perf}
In this experiment, we study the performance of the two sampling strategies discussed in Section \ref{sec:kafka-sampler} with the Intel wireless dataset. We implement samplers with different poll size and measure the time cost of collecting 1 million tuples from Kafka. Essentially, this experiment measures the overhead introduced by transferring the data and calling Kafka API.

Results can be found in Table \ref{tbl:sampler-perf} where the Singleton sampler is represented by the row with \code{pollSize} equals to 1. Rows with \code{pollSize} larger than 1 are sequential samplers which retrieve the entire dataset and sample from each poll.

\begin{table}[h]
\begin{tabular}{ccccc}
pollSize & nPolls & total(ms) & ms/poll & \multicolumn{1}{l}{EquivSingletonSR} \\ \hline
1 & 1000000 & 19000 & 0.019 & --- \\ \hline
10 & 10000 & 4000 & 0.04 & 0.21 \\ \hline
100 & 10000 & 2000 & 0.02 & 0.11 \\ \hline
1000 & 1000 & 1500 & 1.5 & 0.08 \\ \hline
10000 & 100 & 1400 & 14 & 0.075 \\ \hline
100000 & 10 & 1700 & 17 & 0.09 \\ \hline
\end{tabular}
\caption{We use a singleton sampler (pollSize$=$1) and sequential samplers (pollSize$>$1) to sample 1 million tuples, given the latency of each sequential samplers, we derive a equivalent sample rate of the singleton sampler indicate a sample rate above which sequential samplers take less time to collect all requested random samples. If the expected total latency is acceptable, sequential samplers might be preferred in these scenarios.}
\end{table}
\label{tbl:sampler-perf}

Because we have to wait a \code{total(ms)} time for sequential samplers to collect the samples. We calculate the applicable \code{sampleRate} for each sequential sampler above which they can achieve a lower \code{total(ms)}.
For example, if our sample rate is 10\% which is larger than 7.5\% achieved by the best Sequential sampler with pollSize of 10000, the singleton sampler will take more time than the sequential sampler to complete the sampling process. And the sequential samplers might be preferred if the estimated latency of \code{total(ms)} is acceptable.

In \sys, because the sample rate we use during initialization is no larger than 1\%, we always use a singleton sampler during initialization, i.e. to collect sample to build the partition tree and to initialized the strata. For the catch-up phase, if our catch-up rate is larger than 10\%, when dealing with a dataset of medium to large size with an acceptable latency, we will prefer to use a sequential sampler because we can have a complete catch-up with a better accuracy, otherwise, a singleton sampler is preferred for the low latency it offers and we will keep it running in background.

\section{Sufficiently Large Strata for Proportional Allocation}
\label{appndx:proportional}
Let $\alpha=m/N$ be the current sampling rate.
We argue that if each stratum has size at least $N_i\geq \frac{16}{\alpha}\log k$, then all of them satisfy the proportional allocation requirements up to a factor of $2$ with probability at least $1-1/k$.

Let $X_i$ be the number of samples in the $i$-th stratum with $N_i\geq \frac{16}{\alpha}\log k$. We have $E[X_i]=\alpha N_i$.
From Chernoff bound\footnote{We use the multiplicative Chernoff bound: $Pr[X\leq (1-\delta)\mu]\leq e^{-\delta^2\mu/2}$, for $\delta\geq 0$, and $\mu=E[X]$.} we have
$$Pr\left[X_i\leq \frac{1}{2}\alpha N_i\right]\leq \frac{1}{e^{\alpha N_i/8}}\leq \frac{1}{e^{2\log k}}=\frac{1}{k^2}.$$
There are $k$ buckets, so every bucket satisfies the proportional allocation requirements, up to a factor of $2$, with probability at least $1-1/k$.

Now, we discuss how easy is to satisfy the constraint $N_i\geq \frac{16}{\alpha}\log(k)$. The size of the dataset $N$ should be $N\geq \frac{16 k\log k}{\alpha}=\frac{16 k \log k}{m}N$. Equivalently, $16 k \log(k)\leq m$. This inequality is always true for large enough datasets. For example, using $N=4000000$ (which is (on the smaller size) a typical size of the real datasets we tried) with $\alpha=1\%$ we have $16 k \log(k)\leq 40000\Leftrightarrow k\leq 303$, which is always the case in our experiments.
More generally, if we set $k\approx \frac{0.25}{100}m$, which is usually the largest value we use in our experiments, then $\log(k)\leq 25$, which is always true in real scenarios.

\section{Variance Estimator Details}
\label{appndx:varest}
Using algebra we have the following formulas:
For a COUNT/SUM query $q$ (for COUNT we assume that $t.a=1$ for any tuple $t$),
$$w_i\cdot mean(\phi_q(\samples_i))=\frac{1}{m_i}\sum_{t\in \samples_i}\phi(t)=\frac{N_i}{m_i}\sum_{t\in \samples_i\cap q}t.a$$
$$w_i\cdot mean(\phi_q(H_i))=\frac{1}{h_i}\sum_{t\in H_i} \phi(t)=\frac{N_i}{h_i}\sum_{t\in H_i}t.a$$
$$w_i^2\frac{var(\phi_q(\samples_i))}{m_i}=\frac{N_i^2}{m_i^3}\left[m_i\sum_{t\in \samples_i\cap q}t.a^2-\left(\sum_{t\in \samples_i\cap q}t.a\right)^2\right]$$
$$w_i^2\frac{var(\phi_q(H_i))}{h_i}=\frac{N_i^2}{h_i^3}\left[h_i\sum_{t\in H_i}t.a^2-\left(\sum_{t\in H_i}t.a\right)^2\right].$$
For an AVG query $q$ we have
$$w_i\cdot mean(\phi_q(\samples_i))=\frac{N_i}{|\samples_i\cap q|\cdot \sum_{i\in \mathcal{I}_q}N_i}\sum_{t\in \samples_i\cap q}t.a$$
$$w_i\cdot mean(\phi_q(H_i))=\frac{N_i}{h_i\cdot \sum_{i\in \mathcal{I}_q}N_i}\sum_{t\in H_i}t.a,$$
where $\mathcal{I}_q$ is the set of all leaf nodes intersected (either partially or fully) by the query $q$.
$$w_i^2\frac{var(\phi_q(\samples_i))}{m_i}=\frac{w_i^2}{m_i\cdot |\samples_i\cap q|^2}\left[m_i\sum_{t\in \samples_i\cap q}t.a^2-\left(\sum_{t\in \samples_i\cap q}t.a\right)^2\right]$$
$$w_i^2\frac{var(\phi_q(H_i))}{h_i}=\frac{w_i^2}{h_i^3}\left[h_i\sum_{t\in H_i}t.a^2-\left(\sum_{t\in H_i}t.a\right)^2\right]$$

\section{Partition algorithms}
\label{appndx:partAlg}

\subsection{Maximum variance under updates}
\label{sec:updateVar}

An important procedure of all our algorithms is to find what is the maximum error of a query that lies completely in a leaf node or more generally in a rectangle. We explain how to do it for the queries COUNT, SUM, AVG.
In particular, we describe a stronger result: We construct a dynamic data structure over a set of samples $\samples$ with efficient update time such that given a query rectangle, it returns an approximation of the variance of the query with the maximum variance in the query rectangle.
This result leads to very efficient dynamic algorithms for checking the maximum variance and re-constructing a new partition, as we will see in the next subsections.
Let $\maxVar(R)$ be the value of the variance returned by our approximation algorithm in a rectangle $R$.

\mparagraph{COUNT queries}
For COUNT queries it is known~\cite{liang2021combining} that the query with the maximum variance in a rectangle $R$ contains exactly $|R\cap \samples|/2$ samples.
Hence, we construct a dynamic range tree $T$ over $\samples$ with space $O(m\log^d m)$. $T$ can be constructed in $O(m\log^d m)$ time and can be updated in $O(\log^d m)$ time. Given a query rectangle $R$, we run a binary search using $T$ to find two rectangles that contain $|R_i|/2$ items. The query runs in $O(\log^d m)$ time.

\mparagraph{SUM queries}
For SUM queries it is known~\cite{liang2021combining} that we can get a $\frac{1}{4}$-approximation of the maximum variance query inside $R$ with the following simple approach:
Find two non-intersecting rectangles $R_1, R_2$ such that $R_1\cup R_2=R$, and $|R_1\cap \samples|=|R_2\cap \samples|=|R\cap \samples|/2$. Then, they compare $\sum_{t\in R_1\cap \samples}t.a^2$ with $\sum_{t\in R_2\cap \samples}t.a^2$ and return the variance of the rectangle with the largest sum of squares. The variance of this rectangle is a $\frac{1}{4}$-approximation of the maximum variance in $R$.
Range trees work for any aggregation function so we can also use them to compute the sum of the values squared in a rectangle or the variance of a rectangle.
Hence, we can use a dynamic range tree as we had in the COUNT case returning a $\frac{1}{4}$-approximation of the maximum variance.
This data structure has exactly the same complexities as the data structure for COUNT queries.


\mparagraph{AVG queries}
For AVG queries the offline algorithms of~\cite{liang2021combining} cannot be efficiently extended to the dynamic case we are interested in (is not known how to achieve $\polylog m$ update time).
Here we propose a new dynamic data structure for finding an approximation of the maximum variance AVG query in a query rectangle efficiently. The new data structure we propose does not only handle updates efficiently, unlike the data structures in~\cite{liang2021combining}, but it also improves the approximation factor for any dimension $d$.

For a set of samples $\samples$ we construct a dynamic range tree $T'$. We also initialize an empty dynamic data structure $T$ that stores weighted rectangles. Given a query rectangle $R$, it returns the rectangle with the highest weight that lies completely inside $R$. For example $T$ can be a dynamic range tree in $2d$ dimensions storing each rectangle as a $2d$ point by creating a point from its two main opposite corners.
Notice that every $d$-th level node $u$ of $T'$ corresponds to a rectangle $R_u$.
If $|R_u\cap \samples|\leq \delta m$ then we add $R_u$ in $T$ with weight $S(R_u)=\sum_{t\in R_u\cap \samples}t.a^2$. If $\delta m\leq |R_u\cap \samples|\leq 2\delta m$ we split $R_u$ into two rectangles $R_{u_1}, R_{u_2}$ such that $R_{u_1}\cup R_{u_2}=R_u$ and $|R_{u_1}\cap \samples|=|R_{u_2}\cap \samples|=|R_u\cap \samples|/2$. We add $R_{u_1}, R_{u_2}$ in $T$ with weights $S(R_{u_1})=\sum_{t\in R_{u_1}\cap \samples}t.a^2$, $S(R_{u_2})=\sum_{t\in R_{u_2}\cap \samples}t.a^2$.
The tree $T'$ can be constructed in $\tilde{O}(m)$ time and it has $\tilde{O}(m)$ space. In $T$ we might insert $\tilde{O}(m)$ rectangles so it has $\tilde{O}(m)$ space and can be constructed in $\tilde{O}(m)$ time. For any insertion or deletion of a sample in $\samples$, $T'$ can be updated in $\tilde{O}(1)$ time by modifying at most $\tilde{O}(1)$ nodes. For each modified node we update accordingly the corresponding rectangle in $T$ in $\tilde{O}(1)$ time. Furthermore, after an update we traverse the $d$-th level of $T'$ from the updated leaf nodes to their roots inserting or removing rectangles from $T$ accordingly based on the number of points they contain.
Using~\cite{bentley1980decomposable, overmars1981worst, erickson2011static} we can propose a simple dynamic data structure with amortized update time guarantee that can be extended to worst case guarantee by standard techniques~\cite{erickson2011static}.
Overall our data structure can be constructed in $O(m\log^{3d} m)$ time, has $O(\log^{3d} m)$ space and can be updated in $O(\log^{3d+1} m)$ time.


Given a query rectangle $R$ such that $R\cap \samples>2\delta m$ (as we had  in~\cite{liang2021combining}) we show how to return an approximation of the maximum variance query  efficiently.
We search $T$ using the query rectangle $R$ and we get a set of $\tilde{O}(1)$ canonical nodes that contain all rectangles completely inside $R$. From the canonical subsets we get the rectangle $q'$ inside rectangle $R$ with the largest weight.
If $|q'\cap \samples|<\delta m$ then 
using $T'$ we run binary search over all dimensions until we find an expansion of $q'$ that contains exactly $\delta m$ samples. This can be done in $\tilde{O}(1)$ time. Without loss of generality assume that $q'$ contains exactly $\delta m$ samples. Using $T'$ we measure the variance of $q'$ in $R$ in $\tilde{O}(1)$ time. In the end we return the variance of $q'$ as the approximation of the maximum variance AVG query in the query rectangle $R$. The query procedure takes $\tilde{O}(1)$ time.

\begin{lemma}
It holds that $\nu_s(q')\geq \frac{1}{4\log^{d+1} m}\optVar(R)$.
\end{lemma}
\begin{proof}
Let $q$ be the AVG query (a rectangle) with the maximum variance in $R$. It is known from~\cite{liang2021combining} that $q$ contains at most $2\delta m$ and at least $\delta m$ samples.
Since $|R\cap \samples|\geq 2\delta m$ we have that $$|R\cap \samples|\sum_{t\in q'\cap \samples}t.a^2-\left(\sum_{t\in q'\cap \samples}t.a\right)^2\geq \frac{|R\cap \samples|}{2}\sum_{t\in q'\cap \samples}t.a^2,$$
following from Lemma A.2 in~\cite{liang2021combining}.
Next, notice that a query procedure on $T'$ with the query range $q$ would give a set of $\log^d m$ canonical rectangles that cover $q$ where each of them contains at most $2\delta m$ samples. We note that for each of these canonical rectangles in $T'$ there are at most two rectangles in $T$ containing the same items. Let $X$ be the set of rectangles in $T$ corresponding to the canonical rectangles in $T'$. From its definition it holds that $|X|\leq 2\log^{d+1} m$. All of these rectangles in $X$ lie completely inside $R$ so the query procedure will consider them to find the rectangle with the largest weight.
Hence it holds that 
$$\sum_{t\in q'\cap \samples}t.a^2\geq \max_{x\in X}\sum_{t\in x\cap \samples} t.a^2\geq \frac{1}{2\log^{d+1} m}\sum_{t\in q\cap \samples}t.a^2.$$ Overall we have, (for simplicity $|R|=|R\cap S|$)
\begin{align*}
    \nu_s(q')&=\frac{1}{|R|\cdot|q'\cap \samples|^2}\left[|R\cap \samples|\sum_{t\in q'\cap \samples}t.a^2-\left(\sum_{t\in q'\cap \samples}t.a\right)^2\right]\\
    &\geq \frac{1}{|R|\cdot|q'\cap \samples|^2}\frac{|R\cap \samples|}{2}\sum_{t\in q'\cap \samples}t.a^2\\
    &\geq \frac{1}{|R|\cdot|q'\cap \samples|^2}\frac{1}{4\log^{d+1} m}|R\cap \samples|\sum_{t\in q\cap \samples}t.a^2\\
    &\!\geq\!\frac{|q\cap \samples|^2}{|q'\cap \samples|^2}\frac{1}{4\log^{d+1} m}\frac{1}{|R|\!\cdot\!|q\cap \samples|^2}\left[|R\cap \samples|\!\!\!\sum_{t\in q\cap \samples}\!\!t.a^2\!-\!\left(\!\sum_{t\in q\cap \samples}\!\!t.a\!\right)^2\!\right]\\
    &\geq \frac{1}{4\log^d m}\optVar(R)
\end{align*}
\end{proof}

If $d=1$ we can modify the data structure so that it gives an approximation factor $4$.

Overall, given a set $\samples$ of $m$ points in $d$ dimensions we can construct a data structure of space $O(m\cdot \polylog (m))$ in $O(m\cdot \polylog (m))$ time with update time $O(\polylog (m))$ such that given a query rectangle it finds an approximation of the maximum variance COUNT, SUM or AVG query inside $R$ in $O(\polylog (m))$ time.

\subsection{Partition for $d=1$}
\label{appndx:1dpartition}

For COUNT queries the optimum partition in $1$D consists of equal size buckets (intervals) so we can find the new partition in $O(k\log m)$ time by maintaining the order of the samples $\samples$ under insertion or deletion using a balanced search binary tree where the samples are stored in the leaf nodes. Such a tree can be updated in $O(\log m)$ time while the order of the samples on the real line is the same as the order of the leaf nodes from left to right. When we have to (re-)construct the partition we find the right endpoint of each bucket using the search binary tree. Overall, we need $O(\log m)$ time to update the tree and $O(k\log m)$ time to construct a new partition.

Next, we focus on SUM and AVG queries.

\mparagraph{Bounding the error}
We first show a lemma that bounds the maximum length of the largest possible confidence interval among queries that intersect one bucket of the partition.
We assume that the value of any item in $\fullData$ is bounded by a maximum value $\mathcal{U}$ and a minimum non-zero value $\mathcal{L}$. We allow items to take zero values since this is often the case in real datasets but no item with positive value less than $\mathcal{L}$ or larger than $\mathcal{U}$ exists. We assume that $\mathcal{U}=O(\poly(N))$ and $\mathcal{L}=\Omega(1/\poly(N))$.

\begin{lemma}
\label{lem:bound}
Let $R$ be any rectangle and
let $\optVar_S(R), \optVar_A(R)>0$ be the variance of the SUM and AVG query respectively with the maximum variance in $R$. Then it holds that $\frac{\mathcal{L}}{\sqrt{2}}\leq \sqrt{\optVar_S(R)}\leq N\mathcal{U}$ and $\frac{\mathcal{L}}{\sqrt{2}N}\leq \sqrt{\optVar_A(R)}\leq \sqrt{N}\mathcal{U}$.
\end{lemma}
\begin{proof}
Without loss of generality let $q$ be the SUM or AVG query with the maximum variance in $R$.

First we focus on SUM queries. Let $N_R=|R\cap \fullData|$ be the number of total tuples in $R$. Unless $\nu_s(q)=0$, from~\cite{liang2021combining}, we know that there exists a query $q'$ with $|q'|=|R|/2$ such that $\nu_s(q)\geq \frac{N_R^2}{|R|^3}\frac{|R|}{2}\sum_{t\in q'}t.a^2$, where $\sum_{t\in q'}t.a^2>0$. We also have $\sum_{t\in q'}t.a^2\geq \mathcal{L}^2$ (and $w_u=1$) leading to $\sqrt{\optVar_S(R)}\geq\frac{N_R}{\sqrt{2}|R|}\mathcal{L}\geq \frac{\mathcal{L}}{\sqrt{2}}$. Furthermore, we have $\nu_s(q)\leq\frac{N_R^2}{|R|^2}|R|^2\mathcal{U}^2\leq N^2\mathcal{U}^2$ leading to $\sqrt{\optVar_S(R)}\leq N\mathcal{U}$.

Next, we consider AVG queries. Unless $\nu_s(q)=0$, from~\cite{liang2021combining}, we know that there exists a query $q'$ with $|q'|=\delta m\leq |R|/2$ such that $\nu_s(q')\geq \frac{1}{|R|\delta^2m^2}\frac{|R|}{2}\sum_{t\in q'}t.a^2$, where $\sum_{t\in q'}t.a^2>0$. We also have $\sum_{t\in q'}t.a^2\geq \mathcal{L}^2$ (and $w_u=1$) leading to
$\sqrt{\optVar_A(R)}\geq \frac{1}{\sqrt{2}\delta m}\mathcal{L} \geq \frac{\mathcal{L}}{\sqrt{2}N}$.
Furthermore, we have $\nu_s(q)\leq\frac{|R|}{|q|^2}\mathcal{U}^2\leq |R| \mathcal{U}^2\leq N \mathcal{H}^2$ leading to $\sqrt{\optVar_A(R)}\leq \sqrt{N}\mathcal{U}$.
\end{proof}

Since, $\mathcal{U}, \mathcal{L}$ are bounded by a polynomial with respect to $N$, we have that the length of the longest confidence interval is bounded by $O(\poly(N))$ and $\Omega(1/\poly(N))$, i.e. $$\Omega(1/\poly(N))\leq \sqrt{\optVar_S(R)}, \sqrt{\optVar_A(R)} \leq O(\poly(N)).$$

\mparagraph{Description of algorithm}
We describe the partition algorithm for SUM queries. The procedure is identical for AVG queries and we highlight the differences in the end of this section.
For a parameter $\rho\in \mathbb{R}$ with $\rho>1$, let $E=\{\rho^t\mid t\in  \mathbb{Z}, \frac{\mathcal{L}}{\sqrt{2}}\leq \rho^t\leq N\mathcal{U}\}\cup\{0\}$, be the discretization of the range $[\frac{\mathcal{L}}{\sqrt{2}}, N\mathcal{U}]$, i.e., the lower and upper bound of the longest confidence interval (assuming queries completely inside one bucket), by the multiplicative parameter $\rho$.
For an interval $b$, let $\MAXV(b)$ be the approximation of the query with the maximum variance in bucket $b$ (supporting updates) as described in Section~\ref{sec:updateVar}.
We run a binary search on the values of $E$. For each value $e\in E$ we consider, we try to construct a partition of $k$ buckets such that in each bucket the length of the longest confidence interval is at most $e$. If there exists such a partition we continue the binary search with values $e'<e$. If there is no such a partition we continue the binary search with values $e'>e$. In the end of the binary search we return the last partition that we were able to compute.

It remains to describe how to check if a partition with $k$ buckets (intervals) with maximum length confidence interval at most $e$ exists. We start with the leftmost sample, say $t_1$, which is the left boundary of the first bucket. In order to find its right boundary we run a binary search on the samples $\samples$. Let $t_j$ be one of the right boundaries we check in the binary search, and let $b_1=[t_1, t_j]$. If $\sqrt{\MAXV(b_1)}<e$ then we continue the binary search with a sample at the right side of $t_j$ (larger bucket). Otherwise, we continue the binary search with a sample at the left side of $t_j$ (smaller bucket). When we find the maximal bucket with longest confidence interval at most $e$ we continue with the second bucket repeating the same process for at most $k$ buckets. In the end, if all samples in $\samples$ are contained in $k$ buckets then we return that there exists a partition (with $k$ buckets) with maximum variance at most $e$. If we cannot cover all samples in $k$ buckets then we return that there is no partition (with $k$ buckets) with maximum variance at most $e$.

The same algorithm also works for AVG queries. The only difference is that $E$ should be defined with respect to the upper and lower bound of the longest confidence interval, as shown in Lemma~\ref{lem:bound}.

\mparagraph{Correctness}
Before we start with the correctness proof of our algorithm we recall that in~\cite{liang2021combining} we showed that under a mild assumption, for two buckets $b_i, b_j$ if $b_i\subseteq b_j$ then $\sqrt{\optVar(b_i)}\leq \sqrt{\optVar(b_j)}$, namely the length of the longest confidence interval in $b_i$ is smaller than the length of the longest confidence interval in $b_j$. This is the monotonic property of the longest confidence interval.

We assume that $\MAXV(b_i)$ computes a $\frac{1}{\gamma}$-approximation of the maximum variance in $b_i$, i.e., $\MAXV(b_i)\geq \frac{1}{\gamma}\optVar(b_i)$.
Let $\mathcal{R}^*$ be the optimum partition 
and let $b^*$ be the bucket that contains the query with the longest confidence interval in $\mathcal{R}^*$.
First, we notice that if $e\geq \sqrt{\optVar(b^*)}$ then we always find a partition with longest confidence interval at most $e$.
We can show it by induction on the right boundaries of the buckets (intervals) and the monotonic property of confidence intervals.
For the base case,  let $b_1^*=[t_1, t_2]$ be the first bucket of partition $\mathcal{R}^*$. The procedure $\MAXV$ always underestimates the maximum variance in an interval so the binary search in our procedure will consider the right boundary to be greater than $t_2$.
Let $t_i$ be the right boundary of the $i$-th bucket in $\mathcal{R}^*$ and let assume that the $i$-th bucket in our procedure has a right boundary $t_j\geq t_i$.
We consider the $(i+1)$-th bucket in $\mathcal{R}^*$ with boundaries $[t_{i+1}, t_r]$. We show that the $(i+1)$-th bucket in our procedure has a right boundary at least $t_r$. Let $[t_a, t_b]$ be the boundaries of the $(i+1)$-th bucket in our procedure. We have $t_a\geq t_{i+1}$. If $t_a=t_{i+1}$ then $t_b\geq t_r$ as in the basis case. If $t_a>t_{i+1}$ then because of the monotonic property of the confidence intervals and the fact that the $\MAXV$ procedure underestimates the maximum variance we also have that $t_b\geq t_r$.
Let $e'$ be the smallest value in $E$ such that $\sqrt{\optVar(b^*)}\leq e'$. Because of the previous observation our algorithm always returns at least a valid partition for an $e\leq e'$. For every bucket $b$ of this partition, $\sqrt{\MAXV(b)}\leq e$. Let $b'$ be the bucket in the returned partition containing the query with the longest confidence interval in the partition. We have, $\sqrt{\optVar(b')}\leq \sqrt{\gamma \MAXV(b')}\leq \sqrt{\gamma} e\leq \sqrt{\gamma} e'\leq \rho\sqrt{\gamma} \sqrt{\optVar(b^*)}$.

From Section~\ref{sec:updateVar} we have that $\gamma=4$ for SUM and AVG queries queries. So we get a partition where the maximum error is within $2\rho\sqrt{2}$ of the optimum error for SUM queries and within $2\rho$ of the optimum error for AVG queries.

\mparagraph{Running time}
We assume that $\MAXV(\cdot)$ can be computed in $M$ time. Since, $\mathcal{L}, \mathcal{U}$ are polynomially bounded on $N$ we have that $|E|=O(\log_\rho N)$ and it can be constructed in $O(\log_\rho N)$ time. The binary search over $E$ takes at most $O(\log \log_\rho N)$ steps. For each value $e\in E$ of the binary search we check if there is a partition with $k$ buckets and longest confidence interval at most $e$. For each possible bucket we run a binary search over the samples $\samples$ and we run the procedure $\MAXV$ to get an approximation of the maximum variance. Hence, we can decide if there exists a partition with confidence interval $e$ in $O(kM\log m)$ time. Overall, our algorithm takes $O(kM\log m \log\log_\rho N)$. If $\rho$ is a constant, for example $\rho=2$ then the running time is $O(kM\log m \log\log N)$. From Section~\ref{sec:updateVar} we have that in $1$-dimension $M=O(\log m)$ for SUM and $M=O(\log^2 m)$ for AVG queries. Notice that if we skip the $\log$ factors the running time depends only linearly on the number of buckets $k$.

\subsection{Partition for higher dimensions}
We construct a partition by building a k-d tree using the dynamic procedures $\MAXV$ as shown in Section~\ref{sec:updateVar}.
Using the results of~\cite{liang2021combining} we could construct a near optimum k-d tree in time $O(km)$ skipping the $\log$ factors. Here, we use our new results from Section~\ref{sec:updateVar} to construct a k-d tree faster (roughly $O(k)$) with better approximation approximation guarantees.

We start by constructing a dynamic data structure from Section~\ref{sec:updateVar} over the initial set of samples $\samples$.
Assume that after a number of insertions and deletions in $\samples$ we want to (re)construct the tree structure $\PASStree$ over $\samples$.
We construct a partition tree on $\samples$ using ideas from the balanced k-d tree construction.
We pre-define an ordering of the dimensions.
Each node $u$ of the tree is associated with a rectangle $R_u$.
We build the tree in $k$ iterations in a top-down manner starting from the root $v$ such that $\fullData\subset R_v$. 
In any iteration we store and maintain the approximate maximum variance queries of every leaf node in a max heap $C$. In the end of the $i$-th iteration we have a tree of $i$ leaf nodes. Let $u$ be the leaf node with the maximum $\maxVar(R_u)$ value in $C$. We remove its value from $C$. We find the medium coordinate with respect to the next dimension in the ordering (in this branch of the tree) among the samples $R_u\cap \samples$. We split $R_u$ on the median into two rectangles $R_{u_1}, R_{u_2}$ such that $R_{u_1}\cup R_{u_2}=R_u$ and we construct the children $u_1, u_2$ of the parent node $u$. Using the algorithms from Section~\ref{sec:updateVar} we compute $\maxVar(R_{u_1}), \maxVar(R_{u_2})$ and we insert their values in the max-heap $C$. We continue with the same way until we construct a tree $\PASStree$ with $k$ leaf nodes (buckets).

As we showed in~\cite{liang2021combining} such a tree construction returns a partition which is near optimal with respect to the optimum partition tree construction following the same splitting criterion: split on the median of the leaf node with the largest (real) maximum variance query. In our case we do not always split the nod with the real largest error since we use the approximation function $\maxVar(\cdot)$.

For any query our data structure from Section~\ref{sec:updateVar} can be updated in $\tilde{O}(1)$ time.
Given a (re-)partition activation query over a set $\samples$ of $m$ samples we can construct a new $\PASStree$ with the following guarantees: For SUM queries, $\PASStree$ can be constructed in $O(k\log^d m)$ time with approximation factor $2\sqrt{k}$.
For COUNT queries we get the same construction time $O(k\log^d m)$ but the tree we construct is optimum (with respect to the partition tree with same split criterion).
For AVG queries, $\PASStree$ can be constructed in $O(k\log^{2d} m)$ time with approximation factor $2\log^{d/2} m$.
In all cases we can construct near-optimum partitions in $\tilde{O}(k)$ time.

\section{Re-Partitioning Triggers}
\label{appx:re-partition}
A key contribution of \sys is continuous re-optimization of  the partitioning.
We describe how \sys tracks the variances of the current partitions and decides when to re-partition. We also propose two ways to re-partition: partial or full re-partitioning.

Assume that the current partitioning is $\mathcal{R}$ and let $\maxVar(\mathcal{R})$ be the (approximate) maximum variance query with respect to the current set of samples $\samples$.
The automatic procedure first checks the number of samples in each bucket (leaf node) of the current $\PASStree$. If there is a leaf node $i$ associated with partition $R_i$ such that $|\samples_i|<<\frac{1}{\alpha}\log m$ (recall that $\alpha$ is the sampling rate) then there are not enough samples in $u$ to make robust estimators. Hence, we need to find a new re-partition of $\samples$. Even if the number of samples in each bucket is large our system might enable a re-partition:
For a partition $R_i$ in the leaf node layer of $\PASStree$ let $\maxVar_i=\maxVar(R_i)$ be the (approximate) maximum variance at the moment we constructed $\PASStree$. Let $\beta>1$ be a parameter that controls the maximum allowable change on the variance. It can either be decided by the user or we can set it to $\beta=10$. Assume that an update occurred in the leaf node associated with the partition $R_i$. After the update we run the function $\maxVar_i'=\MAXVar(R_i)$ and we update $\maxVar(\mathcal{R})$ if needed. If $\frac{1}{\beta}\maxVar_i\leq \MAXVar_i'\leq \beta\MAXVar_i$ then the new maximum variance in partition $b_i$ is not very different than before so we do not trigger a re-partition. Otherwise, the maximum variance in bucket $b_i$ changed by a factor larger than $\beta$ from the initial variance $\maxVar_i$.
In this case a re-partition might find a new tree with smaller maximum error.
We compute a new partitioning $\mathcal{R}'$ and hence a new tree $\PASStree$.
If $\maxVar(\mathcal{R}')<\frac{1}{\beta}\maxVar(\mathcal{R})$ then we activate a re-partition restarting the catch-up phase over the new tree $\PASStree$. On the other hand, if $\maxVar(\mathcal{R}')\geq \frac{1}{\beta}\maxVar(\mathcal{R})$ then our current partitioning $\mathcal{R}$ is good enough (its worst error is close to the optimum one) so we can still use it.
Of course, the user can also manually trigger re-partitioning. For example, the user can choose to re-partition once every hour, day, or after $\tau$ insertions and deletions have occurred.

Next, we propose two ways to re-partition the index. In particular, the user can select either partial re-partitioning or full re-partitioning.
full re-partitioning is easy; using the algorithms from the previous subsections we can construct a new partitioning and a new tree structure in near-linear time with respect to the samples. Hence, we focus on partial re-partitioning. Instead of re-partitioning the entire space we can only re-partition the area around the "problematic" leaf node. Let $b_i$ be this leaf node with high error or small number of samples. In order to define the neighboring area around $b_i$ we propose either a predefined way or an automatic way. In both cases, the neighboring area is defined by a parameter $\psi$, which is the level of the tree above $b_i$ that the tree needs to be updated.
In the predefined way, the parameter $\psi$ is a known parameter. We find the node $v$ which is defined as an ancestor of the leaf node $b_i$, $\psi$ levels above $b_i$. Let $\PASStree_u$ is the subtree with root node $u$ and let $l_u$ be the number of leaf nodes in $\PASStree_u$. Using the algorithms from the previous subsections we find a near optimum partition starting from node $u$ with $l_u$ leaf nodes. The running time is near-linear with respect to the samples stored $\PASStree_u$. In the automatic way, we do not know the parameter $\psi$ upfront so we try different values of $\psi$ running a binary search on the levels of the tree until we find a partition with low enough error. For each different value of $\psi$ we try, we run the same partial re-partitioning algorithm as in the static case starting from the node $u$ we are considering in the binary search.

Generally, partial re-partitioning is faster than the full re-partitioning since it only suffices to find a better partitioning in a small area of the space. Furthermore, in partial re-partitioning we can still keep all the current estimations in all nodes of $\PASStree\setminus \PASStree_u$, i.e., the nodes of the tree that are not changed. Hence, the error of queries after a partial re-partitioning is also lower than the error of the queries immediately after a full re-partitioning. However, in both cases we need to restart the catch-up phase over the new tree in order to get good estimators to the nodes that were changed by the partial re-partitioning.
Recall that we cannot get samples from a particular area (ideally samples that stored in the leaf nodes of $\PASStree_u$) hence we run the catch-up phase getting samples from the entire space. We finally note that while the catch-up phase considers samples from the entire space, we only use these samples to improve the estimators in the nodes that are still under-represented, i.e., the catch-up phase time threshold for these nodes has not been completed.

\end{document}